\documentclass[12pt]{article}
\usepackage{amssymb}
\usepackage{graphicx}
\usepackage{enumerate}
\usepackage{algorithm}
\usepackage{algpseudocode}
\usepackage{graphicx}
\usepackage{subcaption}
\usepackage{amsmath}
\usepackage{amsthm}
\usepackage{multirow}
\usepackage{amsfonts}
\usepackage{bm}
\usepackage{tikz}
\usepackage{color}
\usepackage{authblk}
\usepackage[colorlinks=true]{hyperref}
\hypersetup{
	linkcolor= blue, % internal links
	citecolor= blue % citations
}
\usepackage{pgfplots}
\pgfplotsset{compat=newest}

\newcommand{\blind}{1}

\newcommand{\fm}{\mathfrak{m}}

\newcommand{\Tr}{\operatorname{Tr}}

\let\svthefootnote\thefootnote
\newcommand\freefootnote[1]{%
  \let\thefootnote\relax%
  \footnotetext{#1}%
  \let\thefootnote\svthefootnote%
}

\newtheorem{theorem}{Theorem}[section]

\newtheorem{remark}{Remark}[section]
\newtheorem{lemma}{Lemma}[section]
\newtheorem{assumption}{Assumption}[section]
\newtheorem{definition}{Definition}[section]
\newtheorem{corol}{Corollary}[section]

\numberwithin{equation}{section}
\numberwithin{algorithm}{section}

\def\bx{{\bf x}}
\def\by{{\bf y}}
\def\bz{{\bf z}}

\date{}

\begin{document}

	\def\spacingset#1{\renewcommand{\baselinestretch}%
		{#1}\small\normalsize} \spacingset{1}

	%%%%%%%%%%%%%%%%%%%%%%%%%%%%%%%%%%%%%%%%%%%%%%%%%%%%%%%%%%%%%%%%%%%%%%%%%%%%%%
	
	\if1\blind
	{
		\title{\bf Two sample test for covariance matrices in ultra-high dimension}
		\author{Xiucai Ding\thanks{Email: xcading@ucdavis.edu. XCD is an assistant professor at UC Davis. The work is partially supported by NSF DMS-2113489 and UC Davis Office of Research via a grant from COVID-19 Research Accelerator Funding Track.}\hspace{.2cm}, Yichen Hu \thanks{Email: ethhu@ucdavis.edu. YCH is a PhD student at UC Davis under the supervision of XCD.}, and
			Zhenggang Wang \thanks{Email: zggwang@ucdavis.edu. ZGW is  a postdoc at UC Davis under the support of XCD (via NSF DMS-2113489 and UC Davis Office of Research via a grant from COVID-19 Research Accelerator Funding Track).} \\
			Department of Statistics, UC Davis \\}
			%\freefootnote{The authors are sorted alphabetically. XCD proposed the project, wrote the paper, and conducted the theoretical analysis. YCH conducted the literature review,  finished the numerical simulations, real data analysis and analyzed the power of the method under the supervision of XCD. ZGW proofread the paper and provided some helpful discussions and ideas on Algorithm 3 and Theorem 3.1. }
	
		\maketitle
	} \fi
	
	\if0\blind
	{
		\bigskip
		\bigskip
		\bigskip
		\begin{center}
			{\LARGE\bf Title}
		\end{center}
		\medskip
	} \fi
	
	\bigskip
	\begin{abstract}
		In this paper, we propose a new test for testing the equality of two population covariance matrices in the ultra-high dimensional setting that the dimension is much larger than the sizes of both of the two samples. Our proposed methodology relies on a data splitting procedure and a comparison of a set of well selected eigenvalues of the sample covariance matrices on the split data sets. Compared to the existing methods, our methodology is adaptive in the sense that (i). it does not require specific assumption (e.g., comparable or balancing, etc.) on the sizes of two samples; (ii). it does not need quantitative or structural assumptions of the population covariance matrices; (iii). it does not need the parametric distributions or the detailed knowledge of the moments of the two populations. Theoretically, we establish the asymptotic distributions of the statistics used in our method and conduct the power analysis. We justify that our method is powerful under very weak alternatives. We conduct extensive numerical simulations and show that our method  significantly outperforms the existing ones both in terms of size and power. Analysis of two real data sets is also carried out to demonstrate the usefulness and superior performance of our proposed methodology. An $\texttt{R}$ package $\texttt{UHDtst}$ is developed for easy implementation of our proposed methodology.	   
	\end{abstract}
	
	\noindent%
	
%	\newpage
	\spacingset{1.3} 
	% DON'T change the spacing 1.9!
	
	\section{Introduction}

	Testing the equality of two population covariance matrices is of fundamental importance in  statistical analysis. For two samples $\mathcal{X}:=\left\{\bx_i \in \mathbb{R}^p, 1 \leq i \leq n_1 \right\}$ and $\mathcal{Y}:=\left\{\by_j \in \mathbb{R}^p, 1 \leq j \leq n_2 \right\}$ with population covariance matrices $\Sigma_1$ and $\Sigma_2$, respectively, researchers are interested in testing 
	\begin{equation}\label{nullhypothesis}
		\mathbf{H}_0: \Sigma_1=\Sigma_2.
	\end{equation}	
Testing (\ref{nullhypothesis}) is crucial in multivariate analysis and high dimensional statistics. 
First, the applications of many commonly used multivariate tests or techniques rely on whether the population covariance matrices are identical. For example, if two population covariance matrices are the same, the asymptotic distribution of  Hotelling's $T^2$ statistic, which is used to test the equality of means, will be much easier to compute \cite{MR1990662,yao2015large}. For another example, for classification, in order to correctly apply the tool of linear discriminant analysis (LDA), people need to check the equality of population covariance matrices \cite{MR1990662}.  Second, in many gene expression data analysis,  the two sample test for covariance matrices serves an important role in  understanding, classifying and selecting gene associations across different phenotypes, for example, see \cite{hu2010new, dudoit2002comparison}.

In this paper, we will propose a novel method to test (\ref{nullhypothesis}) in the ultra-high dimensional regime that 	
	\begin{equation}\label{eq_dimension}
		p \asymp n_1^{\alpha_1} \ \text{and} \  p \asymp n_2^{\alpha_2}, \ \text{for some constants} \ \alpha_1, \alpha_2>1 .
	\end{equation} 
	 We first summarize some related results and methodologies in Section \ref{sec_somerelatedresults} and then provide an overview of our approach in Section \ref{sec_summaryofourresults}.

		\subsection{Some related results}\label{sec_somerelatedresults}
In the literature of multivariate statistics, testing (\ref{nullhypothesis}) has been well studied in the low dimensional regime when the dimension $p$ is fixed and the sample sizes go to infinity.  For example, see \cite{MR1990662,oldold,Nagao,manly1987comparison, o1992robust,gupta1984distribution,perlman1980unbiasedness}. A common feature of these statistics is that they are based on likelihood ratios and utilize the eigenvalues of some sample covariance matrices.

On the other hand, the aforementioned methods may lose their validity as the dimension $p$ diverges with the sample sizes.  The main reason is that the likelihood ratio which is essentially a function of the eigenvalues of certain sample covariance matrices,  is no more consistent due to the bias caused by the inconsistency of sample covariance matrices \cite{10.1214/09-AOS694}. Motivated by this and based on results in random matrix theory, in the last decades, various modified or new statistical methodologies have been proposed in the high dimensional and comparable regime that $\alpha_1=\alpha_2=1$ in (\ref{eq_dimension}). We now list but a few in the following. In \cite{schott2007test}, under the assumption that both samples are Gaussian, the author proposed a statistic based on the Frobenius norm of the difference of the sample covariance matrices of the two samples. In \cite{10.1214/09-AOS694,zheng2012central,ZBY,ZLST}, assuming that $p<n_1$ or $p<n_2$ so that the precision matrices exist at least for one sample, the authors proposed several tests based on the eigenvalues of the $F$ matrices. The condition $p<n_1$ or $p<n_2$ is weakened to some extent later in \cite{zhang2017optimal}. We emphasize again that the aforementioned methods are all developed in the comparable regime that $\alpha_1=\alpha_2=1$ in (\ref{eq_dimension}). In fact, they usually assume that $p/n_1 \rightarrow y_1$ and $p/n_2 \rightarrow y_2$ and both $y_1$ and $y_2$ will appear in the asymptotic distributions of their proposed statistics. Moreover, to correctly implement their methodologies, people either need to assume the samples are Gaussian or have prior knowledge of the moments of the entries of the random vectors.

However, much less is touched in the ultra-high dimensional regime (\ref{eq_dimension}) except for a few ones under various additional assumptions. In \cite{srivastava2010testing}, under the assumption that both samples are Gaussian and $n_1 \asymp n_2$ (i.e., $\alpha_1=\alpha_2$ in (\ref{eq_dimension})), the authors proposed a test based some normalized traces of the sample covariance matrices. In \cite{LC}, assuming that $n_1 \asymp n_2$ (i.e., $\alpha_1=\alpha_2$ in (\ref{eq_dimension})), under certain regularity conditions on $\Sigma_1$ and $\Sigma_2,$   the authors proposed a test based on some $U$-statistics which is an unbiased estimator of the Frobenius norm of $\Sigma_1-\Sigma_2.$ Later on, with additional assumption that both $\Sigma_1$ and $\Sigma_2$ are banded, \cite{he2018high} proposed another  $U$-statistics based test but only targeting on the super-diagonal elements of the covariance matrices. 
Finally, in \cite{cai2013two}, under the assumption that $n_1 \asymp n_2$ (i.e., $\alpha_1=\alpha_2$ in (\ref{eq_dimension})), some moment assumptions on the random vectors and certain sparsity assumptions on $\Sigma_1$ and $\Sigma_2$, the authors introduced a test based on the maximum standardized element-wise differences between the sample covariance matrices which can be computationally very expensive. In summary, all the existing methods concerning the ultra-high dimensional setting (\ref{eq_dimension}) require that the sample sizes are comparably large $n_1 \asymp n_2$ (i.e., $\alpha_1=\alpha_2$ in (\ref{eq_dimension})). Moreover, they need to impose some quantitative or structural assumptions on $\Sigma_1$ and $\Sigma_2$ and distributional assumptions on the random vectors.  

Motivated by the above issues, in the current paper, we propose a novel methodology to test (\ref{nullhypothesis}) in the ultra-high dimensional setting (\ref{eq_dimension}). Our approach  does not need  assumptions on the sample sizes $n_1, n_2,$ or  quantitative or structural assumptions on the population covariance matrices. Moreover, we do not require the random vectors to have specific distributions like Gaussian. An overview of our method will be given in Section \ref{sec_summaryofourresults}.

\subsection{An overview of our method}\label{sec_summaryofourresults}

In contrast to the methods developed in \cite{srivastava2010testing,LC, he2018high,cai2013two}, which all directly compare the entries of the sample covariance matrices, our proposed approach utilizes the eigenvalues of the sample covariance matrices. However, if we directly compare all the eigenvalues, it can result in several issues. First, since the values of the sample sizes $n_1$ and $n_2$ are different in general, a direct comparison  can lead to bias especially when their orders are different. For example, if $n_1 \gg n_2,$ in our regime (\ref{eq_dimension}),  the sample $\mathcal{Y}$ will have much fewer nonzero eigenvalues to be considered. Second, and most importantly, as has been demonstrated in \cite{10.1214/09-AOS694,DingWang2023,yao2015large,ZBY}, the distribution of the statistics that utilize all the eigenvalues usually involves more unknown quantities like the first four moments of the random samples and the detailed information of $\Sigma_1$ and $\Sigma_2.$ 

To address the above issues, inspired by the recent developments in random matrix theory \cite{10.1214/20-AIHP1086, DingWang2023}, we only compare  a subset of the eigenvalues of the two sample covariance matrices. This resolves the issue of using too many eigenvalues of one sample covariance matrix. Moreover, as will be seen in Corollary \ref{cor_twosample}, under the null hypothesis, the asymptotic distribution of the statistic is very universal in the sense that it does not require the knowledge of any particular information of the population covariance matrices and the moments of the random vectors.  In fact, as can be seen in Theorem \ref{thm_onesample}, regardless of whether (\ref{nullhypothesis}) holds, only the mean parts encode the information of the population covariance matrices. This further makes it easier to study the power of the statistics which shows that our method can reject the null hypothesis under very weak alternative.   

Our proposed methodology (c.f. Algorithm \ref{alg:bootstrapping}) will be presented in Section \ref{sec_methodology}. It is comprised of three important components. The first one is a data splitting procedure (c.f. Algorithm \ref{alg0}) which divides the data in $\mathcal{X} \cup \mathcal{Y}$ into three parts, denoted as $\mathcal{X}^s, \mathcal{Y}^s$ and $\mathcal{Z}^s$ with the same  size $n$ satisfying (\ref{eq_datasizecontrol}). $\mathcal{X}^s$ and $\mathcal{Y}^s$ are the testing beds and   $\mathcal{Z}^s$ is used to generate some useful quantities for us to choose a subset of the eigenvalues of the sample covariance matrices associated with $\mathcal{X}^s$ and $\mathcal{Y}^s.$ The selection of the subset of the eigenvalues is done via the choice of a location parameter $\gamma$ (c.f. (\ref{eq_median})) and a tuning bandwidth $\eta_0$ (c.f. Algorithm \ref{algo2}). Both parameters can be chosen automatically and we basically only compare the eigenvalues lying within the interval $[\gamma-\eta_0, \gamma+\eta_0]$ which leads to our statistic in (\ref{eq_realstatistic}).  The construction of the above statistic only uses one split data set so that some samples may be omitted. In order to use as much information as possible and stabilize our procedure, our second component of the methodology is to repeat the  the splitting procedure  multiple times. Then instead of using the statistic in (\ref{eq_realstatistic}) once, we generate a sequence of such statistics and construct   a summary statistic called \emph{decision ratio} (c.f. (\ref{eq_decisionrationdefinition})). Such a procedure will reduce the variability of testing (\ref{nullhypothesis}) compared with only one data splitting. The last component of our methodology is to provide a critical value $\delta$ for the decision ratio to suggest whether we should accept or reject the null hypothesis. This will be done by a calibration  procedure (c.f. Algorithm \ref{algo3}) which utilizes the very universal properties of our statistics under the null hypothesis (\ref{nullhypothesis}). 

On the theoretical side, we establish the asymptotic distributions for our statistics and decision ratio in Section \ref{sec_normality}. Moreover, we conduct detailed power analysis for our methodology in Section \ref{sec_poweranalysis} which shows that our proposed method will be powerful under weak alternatives. We test our proposed methodology and compare it with the state-of-the-art methods  \cite{srivastava2010testing,LC, he2018high,cai2013two} on both simulated and two real data sets. The numerical results show that our proposed method outperforms the existing ones.

			\subsection{Organization of the paper}\label{sec_org}
				    
	    The rest  of the paper is organized as follows. In Section \ref{sec_methodology}, we introduce our test procedure.  In Section \ref{sec_theoretical}, we provide the theoretical guarantees for our procedure by establishing the asymptotic distributions of the test statistics and conducting power analysis. In Section \ref{sec_numerical}, we  compare our proposed methodology with some of the existing methods via  Monte Carlo simulations and two real data analysis.  An online supplementary file is enclosed to provide the technical proofs in Section \ref{sec_techinicalproof} and the arguments of tuning parameter selection in Section \ref{suppl_tuningparameterselection}.  An $\texttt{R}$ package $\texttt{UHDtst}$ is developed for easy implementation \footnote{\url{https://github.com/xcding1212/UHDtst}}.    
	    %\vspace*{0.1in}

	\section{Methodology}\label{sec_methodology}

%	\subsection{Our proposed test procedure}
	
In this section, we introduce our proposed methodology. 	

\subsection{Construction of test statistics}
Our first step is the data splitting procedure via random sampling. For some integer $n$ satisfying that
\begin{equation}\label{eq_datasizecontrol}
n<\mathsf{N}, \ \text{where} \ \mathsf{N}:=\min \left\{ \frac{\max\{n_1, n_2\}}{2}, n_1, n_2 \right\},
\end{equation}
we follow Algorithm \ref{alg0} to split the data.

	\begin{algorithm}[!ht]
		\caption{\bf Data splitting}
		\label{alg0}
		\normalsize
		\begin{flushleft}
			\vspace{0pt}
			\noindent{\bf Inputs:} $n,$ the data sets $\mathcal{X}$ and $\mathcal{Y}.$   		
			
			\noindent{\bf Step one:} Randomly sample $n$ data points from $\mathcal{X}$ and $\mathcal{Y},$ denoted as $\mathcal{X}^s$ and $\mathcal{Y}^s,$ respectively.   
			
			\noindent{\bf Step two:} For the data set with more samples, say $\mathcal{X}$ (i.e., $n_1 \geq n_2$), we randomly sample $n$ data points from $\mathcal{X} \backslash \mathcal{X}^s$, denoted as $\mathcal{Z}^s.$

			\noindent{\bf Output:} The split data sets $\mathcal{X}^s, \mathcal{Y}^s$ and $\mathcal{Z}^s.$ 
		\end{flushleft}
	\end{algorithm}
	
Algorithm \ref{alg0} will generate three independent data sets with the same sample size $n$. $\mathcal{X}^s$ and $\mathcal{Y}^s$ are the test beds for (\ref{nullhypothesis}) and $\mathcal{Z}^s$ is the reference data set to generate $\gamma$ as discussed in Section \ref{sec_summaryofourresults}. Based on the split data sets, we now construct our test statistics. Let the sample covariance matrices associated with $\mathcal{X}^s, \mathcal{Y}^s$ and $\mathcal{Z}^s$ be
	\begin{equation*}
	\mathcal{Q}_x=\frac{1}{\sqrt{pn}} \sum_{\bx_i \in \mathcal{X}^s} (\bx_i-\bar{\bx}) (\bx_i-\bar{\bx})^\top, \ \mathcal{Q}_y=\frac{1}{\sqrt{pn}} \sum_{\by_i \in \mathcal{Y}^s} (\by_i-\bar{\by}) (\by_i-\bar{\by})^\top,  
\end{equation*}	 
and 
\begin{equation*}
\mathcal{Q}_z=\frac{1}{\sqrt{pn}} \sum_{\bz_i \in \mathcal{Z}^s} (\bz_i-\bar{\bz}) (\bz_i-\bar{\bz})^\top,
\end{equation*}
where $\bar{\bx}, \bar{\by}$ and $\bar{\bz}$ are their sample means, respectively. Note that the scaling $(pn)^{-1/2}$ is different from the typical $n^{-1}.$ As will be seen in Section \ref{sec_theoretical}, it is employed to address the ultra-high dimensionality (\ref{eq_dimension}).

	 In what follows, we always denote the eigenvalue sequences associated with the sample covariance matrix of $\mathcal{X}^s$ as $\{\lambda_j\}$, those associated with $\mathcal{Y}^s$ as $\{\mu_j\}$ and those associated with $\mathcal{Z}^s$ as $\{\gamma_j\}.$ We also assume that they are in the decreasing order. Denote
	\begin{equation}\label{eq_median}
		\gamma:=\text{Median}\{\gamma_j\}. 
	\end{equation}
For the mollifier 
	\begin{equation}\label{eq_mathcalK}
		\mathcal{K}(x):= \begin{cases}0 & |x| \geq 1.05 \\ 1 & |x| \leq 1 \\ \exp \left(\frac{1}{(0.05)^2}-\frac{1}{(0.05)^2-(x+1)^2}\right) & -1.05<x<-1 \\ \exp \left(\frac{1}{(0.05)^2}-\frac{1}{(0.05)^2-(x-1)^2}\right) & 1<x<1.05\end{cases}
	\end{equation}
	and some  $\eta_0 \equiv \eta_0(n) \ll 1$ chosen from Algorithm \ref{algo2}, we will use the following statistic
	\begin{equation}\label{eq_realstatistic}
		\mathbb{T}:=\mathbb{T}_x-\mathbb{T}_y,
	\end{equation}
	where 
	\begin{equation}\label{eq_2sampleindividualstatdef}
		\mathbb{T}_{x}:=\sum_{j=1}^n \left( \frac{\lambda_j-\gamma}{\eta_0} \right) \mathcal{K} \left( \frac{\lambda_j-\gamma}{\eta_0} \right), \ \  
		\mathbb{T}_{y}:=\sum_{j=1}^n \left( \frac{\mu_j-\gamma}{\eta_0} \right) \mathcal{K} \left( \frac{\mu_j-\gamma}{\eta_0} \right).  
	\end{equation}

We provide a few remarks. First, $\mathcal{K}(x)$ is essentially a smooth version of the indicator function that $\mathcal{I}(x)=1_{\{|x| \leq 1\}}.$ The reasoning behind is mainly technical due to the use of Helffer-Sj{\" o}strand formula; see the discussion around (\ref{eq_T1integral}). Second, as shown in Figure \ref{fig_highlevelglobalandlocal}, according to the definition in (\ref{eq_mathcalK}), heuristically, instead of using all the eigenvalues, the statistics in (\ref{eq_2sampleindividualstatdef}) basically only sum up the eigenvalues inside the small neighborhood $[\gamma-\eta_0, \gamma+\eta_0],$ after being properly scaled and shifted.

	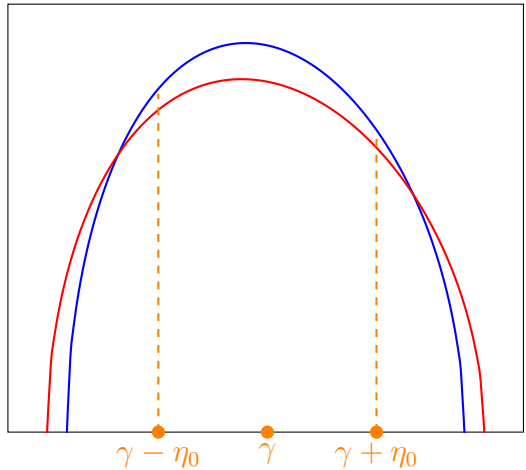
\begin{figure}[htp]
		\centering
		\begin{tikzpicture}
			\begin{axis}[
				title={},
				grid=none,
				ticks=none,
				ymin=0.01
				]
				
				% Define parameters
				\def\aOne{1} % You can modify this value
				\def\bOne{1} % You can modify this value
				\def\r{20} % You can modify this value
				\def\aTwo{1} % You can modify this value
				\def\bTwo{1.1} % You can modify this value
				
				\addplot[blue,	domain=\aOne * \r - 2 * \bOne:\aOne * \r + 2 * \bOne,
				samples=100,  thick,
				smooth] 
				{(\aOne * \r * sqrt((x - (\aOne * \r - 2 * \bOne)) * (\aOne * \r + 2 * \bOne - x))) / (2 * \bOne^2 * pi * x)};
				\addplot[red,	domain=\aTwo * \r - 2 * \bTwo:\aTwo *\r + 2 * \bTwo,
				samples=100,   thick,
				smooth] 
				{(\aTwo * \r * sqrt((x - (\aTwo * \r - 2 * \bTwo)) * (\aTwo * \r + 2 * \bTwo - x))) / (2 * \bTwo^2 * pi * x)};
			\end{axis}
			\fill[orange] (3.45,0) circle (2.5pt) node[align=center, below]{$\gamma$} ; 
			\fill[orange] (2,0) circle (2.5pt) node[align=center, below]{$\gamma-\eta_0$} ; 
			\fill[orange] (4.9,0) circle (2.5pt) node[align=center, below]{$\gamma+\eta_0$} ; 
			\draw[orange,dashed,thick] (2,0) -- (2,4.5);
			\draw[orange,dashed,thick] (4.9,0) -- (4.9,4);
		\end{tikzpicture}
		\caption{Illustration of the statistics $\mathbb{T}_x$ and $\mathbb{T}_y$.}\label{fig_highlevelglobalandlocal}
	\end{figure}

Denote 
	\begin{equation}\label{eq_variancedefinition}
		\mathsf{v}= \frac{1}{2 \pi^2} \int_{\mathbb{R}} \int_{\mathbb{R}} \frac{\left(\mathcal{K}\left(x_1\right)-\mathcal{K}\left(x_2\right)\right)^2}{\left(x_1-x_2\right)^2} \mathrm{~d} x_1 \mathrm{d} x_2. 
	\end{equation}
	We will see later from Section \ref{sec_normality} that under the null hypothesis (\ref{nullhypothesis}), $\mathbb{T}$ will be asymptotically $\mathcal{N}(0, 2 \mathsf{v}).$ Therefore, under the nominal level $\alpha,$ we should reject the null hypothesis if $|\mathbb{T}|>z_{1-\alpha/2},$ where $z_{1-\alpha/2}$ is the $(1-\alpha/2)\%$ quantile of the $\mathcal{N}(0,1)$ random variable.
	
\subsection{Test procedure}	
The construction of the statistics in (\ref{eq_realstatistic}) only uses one split data set so that some samples may be omitted. In order to use as much information as possible and stabilize our procedure, we can repeat the data splitting Algorithm \ref{alg0} and construct the statistics (\ref{eq_realstatistic})  multiple times. Roughly speaking, instead of using (\ref{eq_realstatistic}) once, we use Algorithm \ref{alg0} to generate a sequence of such statistics and construct   a summary statistic called \emph{decision ratio} (c.f. (\ref{eq_decisionrationdefinition})). Such a procedure will reduce the variability of testing (\ref{nullhypothesis}) with only one data splitting. 

 Before stating our inferential procedure, we define the \emph{efficient data splitting} regime.  
\begin{definition}\label{defn_inefficient} We call $\mathcal{X}^s, \mathcal{Y}^s, \mathcal{Z}^s$  from Algorithm \ref{alg0} is an $\epsilon$-efficient data splitting if 
\begin{equation*}
\max\{|\gamma-\mu_1|, |\gamma-\mu_n| \}\leq \operatorname{Range}(\mathcal{Y}^s)-\epsilon \ \text{and} \ \max\{|\gamma-\lambda_1|, |\gamma-\lambda_n|\} \leq \operatorname{Range}(\mathcal{X}^s)-\epsilon, 
\end{equation*}
where $\operatorname{Range}(\cdot)$ is the range of the eigenvalues of the sample covariance matrices of some data set. 
\end{definition}

\begin{remark}\label{rmk_efficientsampling}
We provide a remark on the necessity of introducing Definition \ref{defn_inefficient}. It essentially states that the ranges of the eigenvalues of $\mathcal{Q}_x$ and $\mathcal{Q}_y$ must overlap and $\gamma$ computed from $\mathcal{Q}_z$ should lie in the overlapped interval. The main reason of introducing such a definition is to boost the power of our testing procedure. To ease our explanation, we use Figure \ref{fig_jillustrationofefficientsampling} below for illustrations. In Figure \ref{fig_jillustrationofefficientsampling}, we list three subfigures with some examples for various cases of split data sets.  The red and blue curves represent 
the possible limiting distributions of the eigenvalues of $\mathcal{X}^s$ and $\mathcal{Y}^s$ and the dots show the possible five different locations (in green, black, orange, magenta and cyan) of $\gamma$ computed from $\mathcal{Z}^s.$ 

First, in Figure \ref{subfigurea}, regardless of the locations of $\gamma,$ it does not satisfy Definition \ref{defn_inefficient}. In fact, as can be concluded from Lemma \ref{lem_rigidity}, in this setting, we should reject the null hypothesis in (\ref{nullhypothesis}) even without   doing any hypothesis testing. However, if we directly use our statistic (\ref{eq_realstatistic}) and $\gamma$ happens to be in the orange position, since a small neighborhood of $\gamma$ contains no eigenvalues of $\mathcal{Q}_x$ and $\mathcal{Q}_y$ so that both $\mathbb{T}_x$ and $\mathbb{T}_y$ are small, we may fail to reject (\ref{nullhypothesis}). To address this issue and boost the power, as can be seen from our proposed Algorithm \ref{alg:bootstrapping} (c.f. (\ref{eq_enhancepower})), we will directly reject $\mathbf{H}_0$ if a data splitting like Figure \ref{subfigurea} happens. Second, the data splitting will be efficient if $\gamma$ is in the orange spot in Figure \ref{subfigureb} or black or orange or magenta spot in Figure \ref{subfigurec}. To boost the power, in our Algorithm \ref{alg:bootstrapping}, for cases in Figures \ref{subfigureb} and \ref{subfigurec}, we will only consider efficient splitting (c.f. (\ref{eq_datasplittingalgorithminside})).  

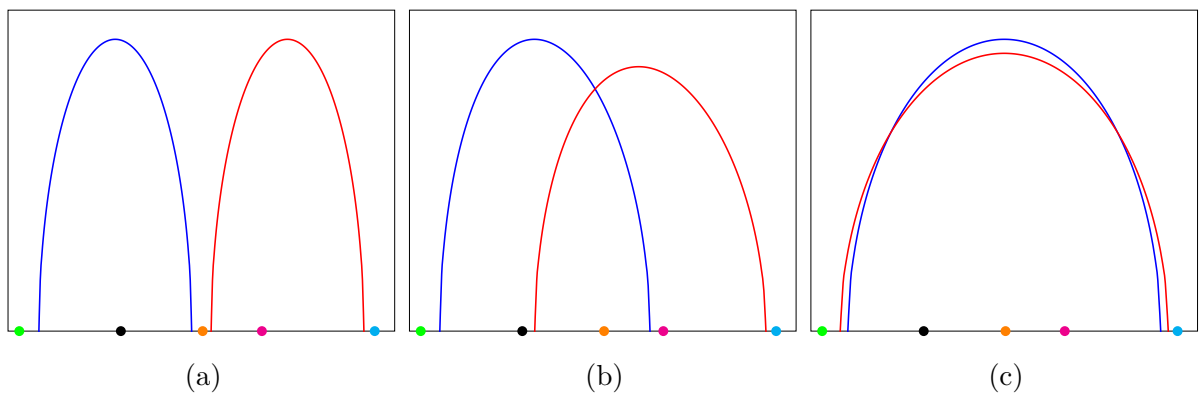
\begin{figure}[hpt]
\begin{subfigure}{0.32\textwidth}
	\begin{tikzpicture}[scale=0.75]
	\begin{axis}[
	title={},
	grid=none,
	ticks=none,
	ymin=0.01
	]
			
			% Define parameters
			\def\aOne{1}
			\def\aTwo{1}
			\def\bOne{1.45}
			\def\bTwo{1}
			\def\r{10}						\addplot[domain=\aOne*\r-2*sqrt(\aTwo):\aOne*\r+2*sqrt(\aTwo), samples=100, smooth, thick, blue] 
			{sqrt(4 * \aTwo - (x - \aOne * \r)^2) / (2 * pi * \aTwo)};
			\addplot[domain=\bOne*\r-2*sqrt(\bTwo):\bOne*\r+2*sqrt(\bTwo), samples=100, smooth, thick,  red] 
			{sqrt(4 * \bTwo - ((x - \bOne * \r)^2)) / (2 * pi * \bTwo)};
		\end{axis}
		\fill[green] (0.2,0) circle (2.5pt);
		\fill[black] (2,0) circle (2.5pt);
		\fill[orange] (3.45,0) circle (2.5pt); 
		\fill[magenta] (4.5,0) circle (2.5pt);
		\fill[cyan] (6.5,0) circle (2.5pt);
	\end{tikzpicture}
	\caption{}\label{subfigurea}
\end{subfigure}
\begin{subfigure}{0.32\textwidth}

	\begin{tikzpicture}[scale=0.75]
	\begin{axis}[
	title={},
	grid=none,
	ticks=none,
	ymin=0.01
	]
			
			% Define parameters
			\def\aOne{1} % You can modify this value
			\def\bOne{1} % You can modify this value
			\def\r{20} % You can modify this value
			\def\aTwo{1.1} % You can modify this value
			\def\bTwo{1.1} % You can modify this value
			
			\addplot[blue,	domain=\aOne * \r - 2 * \bOne:\aOne * \r + 2 * \bOne,
			samples=100,  thick,
			smooth] 
			{(\aOne * \r * sqrt((x - (\aOne * \r - 2 * \bOne)) * (\aOne * \r + 2 * \bOne - x))) / (2 * \bOne^2 * pi * x)};
			\addplot[red,	domain=\aTwo * \r - 2 * \bTwo:\aTwo *\r + 2 * \bTwo,
			samples=100,   thick,
			smooth] 
			{(\aTwo * \r * sqrt((x - (\aTwo * \r - 2 * \bTwo)) * (\aTwo * \r + 2 * \bTwo - x))) / (2 * \bTwo^2 * pi * x)};
		\end{axis}
				\fill[green] (0.2,0) circle (2.5pt);
		\fill[black] (2,0) circle (2.5pt);
		\fill[orange] (3.45,0) circle (2.5pt); 
		\fill[magenta] (4.5,0) circle (2.5pt);
		\fill[cyan] (6.5,0) circle (2.5pt);
	\end{tikzpicture}
		\caption{}\label{subfigureb}	
\end{subfigure}
\begin{subfigure}{0.32\textwidth}
	\begin{tikzpicture}[scale=0.75]
	\begin{axis}[
	title={},
	grid=none,
	ticks=none,
	ymin=0.01
	]
			
			% Define parameters
			\def\aOne{1}
			\def\aTwo{1}
			\def\bOne{1}
			\def\bTwo{1.1}
			\def\r{10}			\addplot[domain=\aOne*\r-2*sqrt(\aTwo):\aOne*\r+2*sqrt(\aTwo), samples=100, smooth, thick, blue] 
			{sqrt(4 * \aTwo - (x - \aOne * \r)^2) / (2 * pi * \aTwo)};
			\addplot[domain=\bOne*\r-2*sqrt(\bTwo):\bOne*\r+2*sqrt(\bTwo), samples=100, smooth, thick,  red] 
			{sqrt(4 * \bTwo - ((x - \bOne * \r)^2)) / (2 * pi * \bTwo)};
		\end{axis}
			\fill[green] (0.2,0) circle (2.5pt);
			\fill[black] (2,0) circle (2.5pt);
		\fill[orange] (3.45,0) circle (2.5pt); 
		\fill[magenta] (4.5,0) circle (2.5pt);
		\fill[cyan] (6.5,0) circle (2.5pt);
	\end{tikzpicture}
		\caption{}\label{subfigurec}
\end{subfigure}
\caption{Possible examples of the LSDs of split data sets $\mathcal{X}^s$ and $\mathcal{Y}^s$ and locations of $\gamma$ from $\mathcal{Z}^s.$}
\label{fig_jillustrationofefficientsampling}
\end{figure}  
\end{remark}

We now propose our two sample test procedure in Algorithm \ref{alg:bootstrapping}. The algorithm can be implemented automatically using our $\texttt{R}$ package $\texttt{UHDtst}.$

%\clearpage	
	\begin{algorithm}[htp]
		\caption{\bf Two sample test procedure}
		\label{alg:bootstrapping}
		
		\normalsize
		\begin{flushleft}
			\vspace{0pt}
			
			\noindent{\bf Inputs:} $n,$ type I error $\alpha,$ and the data sets  $\mathcal{X}$ and $\mathcal{Y}.$
			
			\noindent{\bf Step one:} Run Algorithm \ref{alg0} $\mathsf{K}$ times (say $\mathsf{K}=1,000$) and record the split data sets as $(\mathcal{X}^s_i, \mathcal{Y}_i^s, \mathcal{Z}_i^s),$  whose associated eigenvalues as $(\{\lambda_j^i\}, \{\mu_j^i\}, \{\gamma_j^i\}), 1\leq j\leq n, 1 \leq i \leq \mathsf{K}.$ For $1 \leq i \leq \mathsf{K},$ let the ranges of $\{\lambda_j^i\}, \{\mu_j^i\}$ be as $\mathsf{R}_i(x), \mathsf{R}_i(y),$ respectively.  For $1 \leq i \leq \mathsf{K}$ and a given small positive value $\epsilon_1$ (say $\epsilon_1=0.05$), if
\begin{equation}\label{eq_enhancepower} 
\max\{|\lambda^i_1-\mu^i_n|, |\mu^i_1-\lambda^i_n|\}>\mathsf{R}_i(x)+\mathsf{R}_i(y)+\epsilon_1,
\end{equation}
we record that $\mathsf{c}_i=1$ and denote $\mathcal{S}_0:=\{1 \leq i \leq \mathsf{K}| \ (\ref{eq_enhancepower}) \ \text{is satisfied}\}.$

			\noindent{\bf Step two:}  Compute the median values for $\{\gamma_j^i\}$  as in (\ref{eq_median}) and denote them as $\{\gamma^i\}.$ For some small value $\epsilon$ (say $\epsilon=0.05$), denote the set 
\begin{equation}\label{eq_datasplittingalgorithminside}			
			\mathcal{S}_1:=\{ \{1,2,\cdots, \mathsf{K}\} \backslash \mathcal{S}_0| \text{Definition} \ \ref{defn_inefficient} \ \text{is satisfied} \}. 
\end{equation}			
If $\mathcal{S}_0 \cup \mathcal{S}_1=\emptyset,$ redo Steps one and two until $\mathcal{S}_0 \cup \mathcal{S}_1 \neq \emptyset.$ 					
			
			\noindent{\bf Step three:} For $i \in  \mathcal{S}_1,$  together with $(\{\lambda_j^i\}, \{\mu_j^i\}),$ run Algorithm \ref{algo2} from our supplement to choose a sequence of tuning parameters $\{\eta_0^i \}$. Using the above quantities,  construct a sequence of statistics $\mathbb{T}_i$ following (\ref{eq_realstatistic}) and (\ref{eq_2sampleindividualstatdef}).

			\noindent{\bf Step four:} For $i \in  \mathcal{S}_1$ and the given type one error $\alpha,$  we record 
			\begin{equation*}
				\mathsf{c}_i= \mathbf{1}(|\mathbb{T}_i| \geq z_{1-\alpha/2}\sqrt{2 \mathsf{v}}),  
			\end{equation*}
			where $z_{1-\alpha/2}$ is the $(1-\alpha/2)\%$ quantile of the $\mathcal{N}(0,1)$ random variable.

			\noindent{\bf Step five:} Calculate the decision ratio (DR) 
			\begin{equation}\label{eq_decisionrationdefinition}
				\mathsf{DR}=\frac{1}{|\mathcal{S}_0\cup \mathcal{S}_1|} \sum_{i=1}^\mathsf{|\mathcal{S}_0 \cup \mathsf{S}_1|} \mathsf{c}_i.
			\end{equation}  
			
			\noindent{\bf Output:} Reject the null hypothesis (\ref{nullhypothesis}) if $\mathsf{DR}>\delta.$ Here $\delta>\alpha$ is some control threshold which can be tuned using Algorithm \ref{algo3} from our supplement. 
		\end{flushleft}
	\end{algorithm} 
	
%\clearpage

\begin{remark}\label{rmk_goodgoodgood}
Several remarks are in order on Algorithm \ref{alg:bootstrapping}. First, as discussed in Remark \ref{rmk_efficientsampling}, Steps one and two are mainly employed to increase the power under the alternative. In fact, under the null hypothesis when (\ref{nullhypothesis}) holds, as can be seen from the proof of Corollary \ref{coro_bionomial}, with high probability, (\ref{eq_setdiscuss}) holds.
In other words, (\ref{eq_enhancepower}) and (\ref{eq_datasplittingalgorithminside}) will be skipped and all the split data sets will be used. Second, as will be seen in Corollary \ref{coro_bionomial}, when the null hypothesis (\ref{nullhypothesis}) holds, $\{\mathsf{c}_i\}$ can be asymptotically regarded as a sequence of i.i.d. Bernoulli random variables with probability $\mathsf{p}=\alpha.$ Consequently, when $n$ is sufficiently large, asymptotically, it suffices to check $\mathbf{H}_0:  \mathsf{p}=\alpha \ \operatorname{Vs} \ \mathbf{H}_a: \mathsf{p}>\alpha $ which can be done using the Binomial test or its Gaussian approximation when $\mathsf{K}$ is large. That is, under the nominal level $\alpha$, we need to reject the null hypothesis if  
\begin{equation}\label{eq_theoreticalanalysisofDR}
\mathsf{DR}>\frac{1}{\mathsf{K}} \mathcal{B}_{\mathsf{K}, \alpha}(1-\alpha), \ \ \text{or} \  \ \mathsf{DR}>\alpha+z_{1-\alpha/2} \frac{\sqrt{\alpha(1-\alpha)}}{\sqrt{\mathsf{K}}}, 
\end{equation}
where $ \mathcal{B}_{\mathsf{K}, \alpha}(1-\alpha)$ is the $(1-\alpha)\%$ quantile of a binomial distribution with parameters $\mathsf{K}$ and $\mathsf{\alpha};$ see Corollary \ref{coro_bionomial} for more discussions.  
\end{remark}

\begin{remark}\label{rmk_tuning}
We point out that in order to implement Algorithm \ref{alg:bootstrapping}, several parameters need to be chosen. The first one is the split sample size $n$ which should satisfy (\ref{eq_datasizecontrol}).    In both theoretical and numerical study, we find that the accuracy and powerfulness of Algorithm \ref{alg:bootstrapping} is robust against the choice of $n.$ In our $\mathtt{R}$ package $\texttt{UHDtst}$, for $\mathsf{N}$ in (\ref{eq_datasizecontrol}), we set the default value that $n=\mathsf{N}-5.$ This not only allows us to choose $n$ to be as close as possible to $\mathsf{N}$ but also enables us to choose a reasonably large value of $\mathsf{K}$ in Step one of Algorithm \ref{alg:bootstrapping}. 

The second one is the window size in Step three (c.f. essentially $\eta_0$ in (\ref{eq_2sampleindividualstatdef})). As mentioned before, $\eta_0$ can be regarded as some bandwidth which controls the number of eigenvalues used in the test. Inspired by this aspect, in Algorithm \ref{algo2} of our supplement, we provide some smoothing based approach to choose $\eta_0;$ see Section \ref{suppl_tuningparameterselection} of our supplement for more details. 

The last parameter is the threshold $\delta.$ As discussed in Remark \ref{rmk_goodgoodgood}, according to (\ref{eq_theoreticalanalysisofDR}), our theory has provided some theoretically justified values for $\delta$ when $n$ is sufficiently large. For finite $n$, in order to improve the accuracy and power, in Algorithm \ref{algo3} of our supplement, we provide a calibration procedure  to choose $\delta.$ The motivation is inspired by the results in Corollary \ref{cor_twosample} that when the null hypothesis (\ref{nullhypothesis}) holds, the distribution of (\ref{eq_realstatistic}) only relies on (\ref{eq_variancedefinition}) which is irrelevant of  the matrix $\Sigma$ that $\Sigma_1 = \Sigma_2=\Sigma$. Therefore, we can calibrate a $\delta$ for any combinations of $(n_1, n_2, p, \mathsf{K})$ simply using i.i.d. multivariate Gaussian samples for both $\mathcal{X}$ and $\mathcal{Y};$ see Section \ref{suppl_tuningparameterselection} for more details.       
\end{remark}

%	\subsection{Tuning parameter selection}

	\section{Theoretical guarantees}\label{sec_theoretical}

In this section, we analyze the proposed Algorithm \ref{alg:bootstrapping}. Especially, we
study the asymptotic distributions of our proposed statistics (\ref{eq_realstatistic}) and (\ref{eq_2sampleindividualstatdef}). We will also examine the power of the statistic (\ref{eq_realstatistic}).  For a $k \times k$ symmetric matrix $H$, its empirical spectral distribution (ESD) is defined as $\mu_H:=\frac{1}{k} \sum_{i=1}^k \delta_{\lambda_i(H)}$, where $\delta$ is the Dirac's delta function and $\left\{\lambda_i(H)\right\}$ are the eigenvalues of $H$. For any probability measure $\nu$ defined on $\mathbb{R}$, its Stieltjes transform is defined as
	\begin{equation}\label{eq_definitionStieltjestransform}
	m_\nu(z)=\int \frac{1}{x-z} \mathrm{~d} \nu(x),
	\end{equation}
	where $z \in \mathbb{C}_{+}:=\{E+\mathrm{i} \eta: E \in \mathbb{R}, \eta>0\}$.   For two sequences of deterministic positive values $\{a_n\}$ and $\{b_n\},$ we write $a_n=\mathrm{O}(b_n)$ if $a_n \leq C b_n$ for some positive constant $C>0.$ In addition, if both $a_n=\mathrm{O}(b_n)$ and $b_n=\mathrm{O}(a_n),$ we write $a_n \asymp b_n.$ Moreover, we write $a_n=\mathrm{o}(b_n)$ if $a_n \leq c_n b_n$ for some positive sequence $c_n \downarrow 0.$ Moreover, for a sequence of random variables $\{x_n\}$ and positive real values $\{a_n\},$ we use $x_n=\mathrm{O}_{\mathbb{P}}(a_n)$ to state that $x_n/a_n$ is stochastically bounded. Similarly, we use $x_n=\mathrm{o}_{\mathbb{P}}(a_n)$ to say that $x_n/a_n$ converges to zero in probability.

	\subsection{Some background in random matrix theory}
	
In this section, we provide some background and preliminary results.  Throughout this section, we will need the following mild assumptions. 		
	\begin{assumption}\label{assum_basic}
		We assume that the following conditions are satisfied:
		\begin{enumerate}
			\item For dimensionality, we assume (\ref{eq_dimension}) holds.  
			\item We assume that the two samples are i.i.d. generated according to $\mathbf{x}_i=\Sigma_1^{1/2} \bm{x}_{i} , 1 \leq i \leq n_1,$ and $\mathbf{y}_j=\Sigma_2^{1/2} \bm{y}_{j} \in \mathbb{R}^p, 1 \leq j \leq n_2,$ where the entries of $\bm{x}_i=(x_{is})$ and $\bm{y}_j=(y_{jt}), 1 \leq i \leq n_1, \ 1 \leq j \leq n_2,$ are independent and satisfy that for some positive sequence $(C_k)_{k\in\mathbb{N}}$ that for all $k \in \mathbb{N}$
\begin{equation}\label{eq_onecondition}				
				\mathbb{E} x_{is}=0, \quad \mathbb{E} x_{i s}^2=1, \quad \mathbb{E}\left| x_{i s}\right|^k \leq C_k,
\end{equation}
\begin{equation}\label{eq_twocondition}				
				\mathbb{E} y_{jt}=0, \quad \mathbb{E} y_{j t}^2=1, \quad \mathbb{E}\left|y_{jt}\right|^k \leq C_k. 
\end{equation}							
				\item For $\Sigma_1$ and $\Sigma_2,$ we assume that all of their eigenvalues are bounded from above and below away from zero. 
			\end{enumerate}		
		\end{assumption}
		
\begin{remark}
Several remarks on Assumption \ref{assum_basic} are in order. First, condition 1 specifies the ultra-high dimensional regime. Second, condition 2 provides a data generating model which has been frequently used in high dimensional data analysis, for example, see \cite{chen2010tests,dobriban2019deterministic,he2018high,ke2023estimation,ZBY}. Moreover, we assume the data is centered just for the ease of statements. Our results can be easily generalized to the non-zero mean setting; see Theorem 2.23 of \cite{bloemendal2016principal} for more detailed discussions. In addition, the moment assumptions in (\ref{eq_onecondition}) and (\ref{eq_twocondition}) can be weakened with additional technical efforts as in \cite{ding2018necessary,FYEJP}.  Finally, condition 3 imposes a very mild assumption on the population covariance matrices. 
\end{remark}

For the split data sets $\mathcal{X}^s$ and $\mathcal{Y}^s$ from our Algorithm \ref{alg0}, under Assumption \ref{assum_basic}, we can write the sample covariance matrix as follows	
					\begin{equation}\label{eq_samplecovaraincematrixdefinition}
				\mathcal{Q}_x=\frac{1}{\sqrt{pn}}\Sigma_1^{1/2} XX^\top \Sigma_1^{1/2}, \ \mathcal{Q}_y=\frac{1}{\sqrt{pn}}\Sigma_2^{1/2} YY^\top \Sigma_2^{1/2},
			\end{equation} 
where $X$ contains the samples from $\mathcal{X}^s$ and $Y$ contains that from $\mathcal{Y}^s.$ Since the matrices
		\begin{equation*}
			Q_x=\frac{1}{\sqrt{pn}}X^\top \Sigma_1 X, \ Q_y=\frac{1}{\sqrt{pn}} Y^\top \Sigma_2 Y,
		\end{equation*}  
		have the same non-zero eigenvalues with $\mathcal{Q}_x$ and $\mathcal{Q}_y,$	it is sufficient to work with $Q_x$ and $Q_y.$ It is well-known that the limiting spectral distributions (LSD) of $Q_x$ and $Q_y$ can be best described by their Stileltjes transforms, denoted as $m_1(z)$ and $m_2(z).$ The following lemma characterizes $m_1(z)$ and $m_2(z).$ 
		
		\begin{lemma}\label{lem_density}
			Suppose  $\{\sigma_j^{(i)}\}_{j=1}^p$ is the sequence of the eigenvalues of $\Sigma_i$ and let $\phi:=\frac{p}{n}$. Then for each $z \in \mathbb{C}_{+}$, there exists a unique $m_i \equiv m_i(z) \in \mathbb{C}_{+}, \ i=1,2,$ satisfying
			$$
			\frac{1}{m_i}=-z+ \frac{1}{p}\sum_{j=1}^p\frac{\phi}{\phi^{1 / 2} (\sigma_j^{(i)})^{-1}+m_i}.
			$$
		\end{lemma}
		\begin{proof}
				See Lemma 2.3 of \cite{DingWang2023}. 
		\end{proof}

It is well-known that given $m_i(z),$ people can obtain its associated density function in the sense of (\ref{eq_definitionStieltjestransform}) using the inverse formula \cite{MR2567175} (also see (\ref{eq_inverseformula})). Let $\varrho_i$ be the asymptotic density associated with $m_i$ in Lemma \ref{lem_density}, $i=1,2$. In \cite[Lemma 2.5]{DingWang2023}, it has been also proved that $\varrho_i, i=1,2,$ are both supported on some single intervals that for some constants $\gamma_{\pm}^i$  
 \begin{equation}\label{eq_intervaldefinition}
 \operatorname{supp} \varrho_i \ \cap \ (0, \infty)=[\gamma_-^i, \gamma_+^i], 
\end{equation}  		
where $\gamma_+^{i}-\gamma_-^i=\mathrm{O}(1)$ and $\gamma_-^i, \gamma_+^i \asymp (p/n)^{1/2}$ for $i=1,2.$		
\subsection{Asymptotic distributions of the statistics}\label{sec_normality}

		In this section, we establish the CLTs for the statistics (\ref{eq_realstatistic}) an (\ref{eq_2sampleindividualstatdef}). For $\gamma$ in (\ref{eq_median}),  we define
		\begin{equation}\label{mean}
			\mathsf{M}_x:=n \int_{\mathbb{R}}\frac{t-\gamma}{\eta_0}\mathcal{K}\left(\frac{t-\gamma}{\eta_0}\right)\mathrm{d}\varrho_1(t), \ \ \mathsf{M}_y:=n\int_{\mathbb{R}}\frac{t-\gamma}{\eta_0}\mathcal{K}\left(\frac{t-\gamma}{\eta_0}\right)\mathrm{d}\varrho_2(t).
		\end{equation}
		
		\begin{theorem} \label{thm_onesample}
			Suppose Assumption \ref{assum_basic} holds and $\max\{\gamma_-^1, \gamma_-^2\}<\gamma<\min\{\gamma_+^1, \gamma_+^2\}$. For some small constants $\tau_1, \tau_2>0$ that $n^{-1+\tau_2}<\eta_0 \leq n^{-\tau_1}$, we have that for $\mathbb{T}_x$ and $\mathbb{T}_y$ in (\ref{eq_2sampleindividualstatdef}) and $\mathsf{v}$ in (\ref{eq_variancedefinition}) 
			\begin{equation*}
				\mathbb{T}_x-\mathsf{M}_x \Rightarrow \mathcal{N}(0,\mathsf{v}), \ 		\mathbb{T}_y-\mathsf{M}_y \Rightarrow \mathcal{N}(0,\mathsf{v}). 
			\end{equation*}
			
		\end{theorem}
		\begin{proof}
			See Section \ref{suppl_thmonesampleproof}. 
		\end{proof}
\begin{remark}				
Theorem \ref{thm_onesample} establishes the asymptotic normality for the statistics $\mathbb{T}_x$ and $\mathbb{T}_y.$ We emphasize that the variances of the asymptotic distributions of $\mathbb{T}_x$ and $\mathbb{T}_y$ are identical regardless of whether $\mathbf{H}_0$ in (\ref{nullhypothesis}) holds. The possible differences lie in the mean parts $\mathsf{M}_x$ and $\mathsf{M}_y$ in (\ref{mean}). On the one hand, when (\ref{nullhypothesis}) holds, we have that $\varrho_1 \equiv \varrho_2$ so that $\mathsf{M}_x$ will be sufficiently close to $\mathsf{M}_y$ (c.f. $\mathbb{T}$ will be asymptotically Gaussian with mean zero). On the other hand, when (\ref{nullhypothesis}) fails, $\varrho_i, i=1,2,$ which encode the information of $\Sigma_i$ via Lemma \ref{lem_density} will be different so that $\mathbb{T}_x$ will be different from $\mathbb{T}_y.$ 				
\end{remark}				

For the distribution of $\mathbb{T}$ under the null hypothesis (\ref{nullhypothesis}), due to the independent splitting in Algorithm \ref{alg0}, Theorem \ref{thm_onesample} immediately yields the following result. 
		\begin{corol}\label{cor_twosample}
			Suppose Assumption \ref{assum_basic} holds. Then  under the null hypothesis $\mathbf{H}_0$ in (\ref{nullhypothesis}), for some small constants $\tau_1, \tau_2>0$ that $n^{-1+\tau_2}<\eta_0 \leq n^{-\tau_1}$, we have $$\mathbb{T}\Rightarrow \mathcal{N}(0,2\mathsf{v}).$$
		\end{corol}
		\begin{proof}
			Under these assumptions, it is clear from Lemma \ref{lem_density} and (\ref{eq_intervaldefinition}) that $\varrho_1=\varrho_2 \equiv \varrho, \gamma_-^1=\gamma_-^2 \equiv \gamma_-$ and $\gamma_+^1=\gamma_+^2 \equiv \gamma_+.$ Moreover, by Lemma \ref{lem_rigidity} and the discussion in Remark \ref{rem_rigidity} below, we see that with high probability, $\gamma_-<\gamma<\gamma_+.$ These also imply that $\mathsf{M}_x=\mathsf{M}_y.$ The proof then follows from Theorem \ref{thm_onesample}, the independence splitting in Algorithm \ref{alg0} and $\mathbb{T}=(\mathbb{T}_x-\mathsf{M}_x)-(\mathbb{T}_y-\mathsf{M}_y).$ 
		\end{proof}
		
As a consequence of Corollary \ref{cor_twosample}, we immediately obtain the asymptotic properties of our Algorithm \ref{alg:bootstrapping} in terms of the decision ratio (DR) in (\ref{eq_decisionrationdefinition}).		

\begin{corol}\label{coro_bionomial} Suppose the assumptions of Corollary \ref{cor_twosample} hold. Then when $n$ is sufficiently large, for $\mathsf{DR}$ in (\ref{eq_decisionrationdefinition}) of our Algorithm \ref{alg:bootstrapping}, when the null hypothesis $\mathbf{H}_0$ in (\ref{nullhypothesis}) holds, we have that conditional on the data sets $(\mathcal{X}, \mathcal{Y})$ 
\begin{equation*}
\mathsf{K} \times \mathsf{DR} \Rightarrow \mathcal{B}_{\mathsf{K}, \alpha},
\end{equation*}
where $\mathcal{B}_{\mathsf{K}, \alpha}$ is a Binomial random variable with size $\mathsf{K}$ and probability $\alpha.$
\end{corol}		
\begin{proof}
Analogous to the proof of Corollary \ref{cor_twosample}, by Lemma \ref{lem_rigidity} and the discussion in Remark \ref{rem_rigidity} below, we see that with high probability, for all $1 \leq i \leq \mathsf{K},$
\begin{equation*}
\lambda_1^i=\gamma_++\mathrm{o}(1), \ \mu_1^i=\gamma_++\mathrm{o}(1), \ \lambda_n^i=\gamma_-+\mathrm{o}(1), \ \mu_n^i=\gamma_-+\mathrm{o}(1), 
\end{equation*} 
and 
\begin{equation*}
\gamma_-<\gamma^i<\gamma_+.
\end{equation*}
In view of (\ref{eq_enhancepower}) and Definition \ref{defn_inefficient}, we conclude that with high probability 
\begin{equation}\label{eq_setdiscuss}
\mathcal{S}_0=\emptyset, \ \mathcal{S}_1=\{1,2,\cdots, \mathsf{K}\}.
\end{equation} 
Moreover, according to Corollary \ref{cor_twosample}, when conditional on the data sets, we see that $\{\mathsf{c}_i\}$ are asymptotically i.i.d. Bernoulli random variables with probability $\alpha.$ This immediately completes the proof. 
\end{proof}		

Corollary \ref{coro_bionomial} states that when $n$ is large, we can essentially characterize the asymptotic distribution of $\mathsf{DR}.$  Therefore, as discussed in Remark \ref{rmk_goodgoodgood}, we can  use it as the statistic to test (\ref{nullhypothesis}) (c.f. (\ref{eq_theoreticalanalysisofDR})).

		\subsection{Power analysis}\label{sec_poweranalysis}
		
In this section, we study the power of the statistic $\mathbb{T}$ in (\ref{eq_realstatistic}) under the alternative that 
\begin{equation}\label{hypothesis_alternative}
\mathbf{H}_a: \Sigma_1 \neq \Sigma_2. 
\end{equation}	
For notional simplicity, denote the ESDs of $\Sigma_1$ and $\Sigma_2$ as $\pi_1$ and $\pi_2$ and their associated $k$th moments as 
\begin{equation}\label{eq_emperical}
\mathfrak{m}_k(\Sigma_i)=\int x^k \pi_i(\mathrm{d} x), \ i=1,2.
\end{equation}

We point out that in our Algorithm \ref{alg:bootstrapping}, the statistics (\ref{eq_realstatistic}) will only be used for efficient sampling.  Therefore, we first study the power of the statistics $\mathbb{T}$	for the efficient split data sets.

		\begin{theorem}\label{thm_poweranalysis} Suppose Assumption \ref{assum_basic} holds. Moreover, for some small $\epsilon>0,$ we assume that $(\mathcal{X}^s, \mathcal{Y}^s, \mathcal{Z}^s)$ generated from Algorithm \ref{alg0} is an $\epsilon$-efficient data splitting satisfying Definition \ref{defn_inefficient}.  For some small constants $\tau_1, \tau_2>0,$ we assume that  $n^{-1+\tau_2}<\eta_0 \leq n^{-\tau_1}$. Moreover, suppose that for sufficiently large $n$ and any  constant $\mathsf{c}=\mathrm{O}(n^{-1}+\phi^{-1/2}),$ we have that
		\begin{equation}\label{eq_condition}
		\phi^{1/2}|\mathfrak{m}_1(\Sigma_1)-\mathfrak{m}_1(\Sigma_2)|+\mathsf{c}|\mathfrak{m}_2(\Sigma_1)-\mathfrak{m}_2(\Sigma_2)| \neq 0. 
		\end{equation}
Then given some type I error rate $\alpha,$ suppose the alternative (\ref{hypothesis_alternative}) holds in the sense that 
		\begin{equation}\label{eq_localalternative}
		\phi^{1/2}|\mathfrak{m}_1(\Sigma_1)-\mathfrak{m}_1(\Sigma_2)|+\phi^{-1/2} |\mathfrak{m}_2(\Sigma_1)-\mathfrak{m}_2(\Sigma_2)|>C_{\alpha} \eta_0^{-2} n^{-1} ,
		\end{equation}
where the constant $C_{\alpha} \equiv C_{\alpha}(n) \uparrow \infty$ as $n \rightarrow \infty$, we have that 
\begin{equation}\label{eq_Tpower}
\mathbb{P} \left(|\mathbb{T}|>\sqrt{2 \mathsf{v}} z_{1-\alpha/2} \right)=1.  
\end{equation} 		
		\end{theorem}
		\begin{proof}
		See Section \ref{sec_32proofdetails}. 
		\end{proof}
		
\begin{remark}
A few remarks are in order. First, (\ref{eq_condition})  is a mild condition and can be easily satisfied. In fact, we can actually remove this condition when $\phi^{-1/2} \ll n^{-1},$ or equivalently, $p \gg n^3.$ In such a setting, since $\mathfrak{m}_2(\Sigma_i), i=1,2,$ are bounded from above, we have that $\mathsf{c}|\mathfrak{m}_2(\Sigma_2)-\mathfrak{m}_2(\Sigma_1)|=\mathrm{O}(n^{-1}).$ Consequently, (\ref{eq_localalternative}) implies (\ref{eq_condition}) so that we can remove (\ref{eq_condition}). Second, (\ref{eq_localalternative}) is generally a weak alternative. It suggests that we should use a relative larger $\eta_0$ in order to increase the power. Finally, the condition (\ref{eq_localalternative}) also
relies on the ratio $\phi=p/n.$ It demonstrates that when the ratio $\phi$ increases, we may need weaker alternatives. For example, if $p \gg n^3,$ (\ref{eq_localalternative}) reads as 
\begin{equation*}
|\mathfrak{m}_1(\Sigma_1)-\mathfrak{m}_1(\Sigma_2)|>C_{\alpha} \phi^{-1/2} \eta_0^{-2} n^{-1},
\end{equation*}
which can be much weaker than those used in \cite{cai2013two,LC}. 
\end{remark}

Theorem \ref{thm_poweranalysis}  also yields the results of the power analysis of our Algorithm \ref{alg:bootstrapping} in terms of the decision ratio in  (\ref{eq_decisionrationdefinition}). 		

\begin{corol}\label{coro_convergence}
Suppose Assumption \ref{assum_basic} and (\ref{eq_localalternative}) hold. For  given type I error rate $\alpha,$ when $n$ is sufficiently large, we have that conditional on the data sets $(\mathcal{X}, \mathcal{Y})$ 
\begin{equation*}
\mathsf{DR}=1+\mathrm{o}_{\mathbb{P}}(1). 
\end{equation*}
\end{corol}
\begin{proof}
If $\mathcal{S}_1=\emptyset,$ then by Step one of our Algorithm \ref{alg:bootstrapping}, we have that $\mathsf{DR}=1.$ Otherwise, together with Step four of our Algorithm \ref{alg:bootstrapping} and (\ref{eq_Tpower}), we can see that $\mathsf{DR}=1+\mathrm{o}_{\mathbb{P}}(1).$ This completes our proof. 
\end{proof}
	
Corollary \ref{coro_convergence}	implies that when $n$ is large and the weak local alternative (\ref{eq_localalternative}) holds, $\mathsf{DR}$ will converge to 1 with high probability. Consequently, our Algorithm \ref{alg:bootstrapping} will be able to reject the null hypothesis under the weak alternative as in (\ref{eq_localalternative}) for any threshold $\delta<1$. 
		
		\section{Numerical results}\label{sec_numerical}
In this section, we conduct extensive Monte-Carlo simulations and two real data analysis to show the accuracy and powerfulness of our proposed test procedure Algorithm \ref{alg:bootstrapping}. For better illustrations, we also compare our Algorithm \ref{alg:bootstrapping} (denoted as Proposed) with   four other state-of-the-art methods as in \cite{cai2013two} (denoted as CLX2013), \cite{LC} (denoted as LC2012),\cite{srivastava2010testing} (denoted as SY2010) and \cite{he2018high} (denoted as HC2018). For the user's convenience, all these methods can be implemented using our $\texttt{R}$ package $\texttt{UHDtst}.$  For reproducibility, we also provide an online demo. \footnote{ \url{https://xcding1212.github.io/UHDtst.html}.}   
				
\subsection{Numerical simulations}
In this section, we check and compare the accuracy and power using Monte-Carlo simulations.  		
\subsubsection{Simulation setup}\label{sec_simulationsetup}
As in the second condition of Assumption \ref{assum_basic}, our two samples $\{\mathbf{x}_i\}$ and $\{\mathbf{y}_j\}$ are generated according to $\mathbf{x}_i=\Sigma_1^{1/2} \bm{x}_i$
and $\mathbf{y}_j=\Sigma_2^{1/2} \bm{y}_j,$ where $\bm{x}_i=(x_{is})$ and $\bm{y}_j=(y_{jt})$ satisfy
(\ref{eq_onecondition}) and (\ref{eq_twocondition}). For the i.i.d. entries $x_{is}$ and $y_{jt},$ we consider two different distributions: the standard Gaussian distribution $\mathcal{N}(0,1)$ with vanishing 4th cumulant  and the two-point distribution that $\mathbb{P}(x=\sqrt{2}) = 1/3$ and $\mathbb{P}(x=-\sqrt{2}/2)=2/3$ whose 4th cumulant is $-1.5.$

For the two population covariance matrices $\Sigma_1$ and $\Sigma_2,$ we formulate the null hypothesis $\mathbf{H}_0$ as
\begin{equation}\label{eq_simunull}
\mathbf{H}_0: \Sigma_1=\Sigma_2 \equiv \Sigma^*.
\end{equation}
In the simulations, we will consider three different cases based on (\ref{eq_simunull}) as follows.  	
\begin{enumerate}
\item[(Case I).] We consider model two of \cite{cai2013two}. For $\Sigma^*$ in the null hypothesis (\ref{eq_simunull}), we consider the Toeplitz matrix that $\Sigma^* =(\sigma^{*}_{ij})$, where $\sigma^{*}_{ij}=0.5^{|i-j|}.$ For the alternative (\ref{hypothesis_alternative}), we consider that 
\begin{equation*}
\mathbf{H}_a: \Sigma_1=\Sigma^*, \ \Sigma_2=D^{1/2} \Sigma^* D^{1/2}.
\end{equation*} 
Here $D=\mathrm{diag}(d_{ii})$, where $d_{ii}$'s are generated from $\mathrm{Unif}(0.5,2.5)$. 
\item[(Case II).] We consider case one of \cite{LC}. For $\Sigma^*$ in (\ref{eq_simunull}), we consider $\Sigma^*=I.$ For the alternative, we consider that 
\begin{equation*}
\mathbf{H}_a: \Sigma_1=\Sigma^*, \ \Sigma_2= \Sigma^*+\Delta.
\end{equation*}
Here for some constant $\theta>0, \ \Delta$ is a banded matrix that
\begin{equation*}
\Delta_{ij}=
\begin{cases}
\theta^2 & \text{if} \ i=j, \\
\theta & \text{if} \ |i-j|=1, \\
0 & \text{otherwise}.
\end{cases}
\end{equation*}
In other words, $\Sigma_2$ can be regarded as the covariance matrix of a $p$-dimensional realization of a stationary MA(1) process driven by $\mathcal{N}(0,1)$ random variables with parameter $\theta.$ 
\item[(Case III).] For $\Sigma^*$ in (\ref{eq_simunull}), we set $\Sigma^*=QDQ^\top,$ where $Q$ is some orthogonal matrix and $D=\mathrm{diag}(d_{ii})$ with $d_{ii}$'s being generated from $\mathrm{Unif}(3,6).$ For the alternative, we consider 
\begin{equation}\label{eq_alternativeIII}
\mathbf{H}_a: \Sigma_1=\Sigma^*, \ \Sigma_2=\Sigma^*+\varepsilon I_p, 
\end{equation}
where $\varepsilon>0$ is some constant and $I_p$ is the $p \times p$ identity matrix. 
\end{enumerate}

\subsubsection{Simulation results}		

In this section, we report and discuss the  numerical results based on extensive Monte-Carlo simulations. We conduct the simulations following the settings in Section \ref{sec_simulationsetup}.  For the dimension and sample sizes, we consider $p=6,000$ and various combinations $(n_1, n_2)=(100,100), (100, 150)$, $(100,800), (100,1000).$ We report our results in Tables \ref{caseGaussian} and \ref{caseTwopoint} and Figure \ref{fig_power}. We elaborate our results in more details as follows. 

Tables \ref{caseGaussian} and \ref{caseTwopoint} summarize the results of the empirical size and power of our proposed method in Algorithm \ref{alg:bootstrapping} and the other four methods in the literature \cite{cai2013two,he2018high,LC,srivastava2010testing} for Gaussian samples and two-point samples, respectively.  First, we conclude that across all the simulation settings, our proposed method (i.e., Proposed) is accurate and powerful. It also outperforms all the other methods in terms of both size and power. Second, LC2012 is reasonably accurate for all the simulation settings but lose their power in Case III.
Third, due to the multiple testing procedure, HC2018 is powerful across all the settings but at the expense of being inaccurate. Fourth, SY2010 only works for Case II and is invalid for Cases I and III both in size and power. Finally, when the samples are Gaussian and $n_1$ and $n_2$ are comparably large,  CLX2013 works in Case I and II but loses its power in Case III.  Moreover, if either $n_1$ and $n_2$ are incomparable or the samples follow two-point distribution, CLX2013 will be no more accurate.

Before concluding this section, to show the mildness of the condition (\ref{eq_localalternative}) and the powerfulness of our method,  in Figure \ref{fig_power}, using Case III with the alternative (\ref{eq_alternativeIII}), we report how the simulated power changes with $\varepsilon.$ It can be concluded that our proposed will achieve power one for very small $\varepsilon$ (i.e., weak alternative), while all the other methods either are powerless or require much larger $\varepsilon$ to have  nontrivial power.

		\begin{table}[htp]
			\fontsize{9pt}{7pt}\selectfont
			\centering
			\renewcommand{\arraystretch}{1.7}
					\hspace*{-10pt}
			\begin{tabular}{c|cc|cccc}
				\hline
%				\multicolumn{3}{c|}{}  \\
			%	\cline{4-11}
				\multicolumn{1}{c}{}& & \textbf{Methods/Setting}& $(100,100)$ & $(100,150)$ & (100,800) & (100,1000) \\
				\hline
				\multirow{9}{*}{\rotatebox[origin=c]{90}{\textbf{Empirical Size}}}&\multirow{3}{*}{\textbf{Case I}}& SY2010 & 0 & 0 & 0.007 & 0.017   \\ 
				&&LC2012 & 0.049  & 0.049 & 0.07 & 0.067  \\ 
				&&CLX2013 & 0.038& 0.046 & 0.24 & 0.36  \\
                &&HC2018 & 0.014 & 0.017 & 0.01 & 0.003  \\	
                &&Proposed & 0.045 & 0.047 & 0.051 & 0.048  \\ 	
				 \cline{2-7}
				&\multirow{3}{*}{\textbf{Case II}}& SY2010 & 0.035 & 0.033 & 0.06 & 0.06   \\ 
				&&LC2012 & 0.054 & 0.062& 0.05& 0.05  \\ 
				&&CLX2013 & 0.043& 0.031 & 0.223 & 0.243  \\
                && HC2018 &0.013 & 0.007 & 0.01 & 0 \\	
                && Proposed & 0.048 & 0.049 & 0.052  & 0.05  \\ 	 \cline{2-7}
			&	\multirow{3}{*}{\textbf{Case III}}& SY2010 & 0 & 0 & 0.003 & 0.01    \\ 
				&&LC2012 & 0.048 & 0.049& 0.057 & 0.067  \\ 
				&&CLX2013 & 0.052 & 0.046 & 0.19 & 0.223  \\
                && HC2018 &0.024 & 0.004 & 0.003 & 0.007 \\	
                && Proposed & 0.047 & 0.05 & 0.051 & 0.051  \\ 				
				\hline
				\multirow{9}{*}{\rotatebox[origin=c]{90}{\textbf{Empirical Power}}}&\multirow{3}{*}{\textbf{Case I}}& SY2010 & 0 & 0 & 0.873 & 0.917    \\ 
				&&LC2012 & 1& 1&1 &1  \\ 
				&&CLX2013 & 1& 1& 1 & 1 \\
                && HC2018 &1 &1 & 1 &1  \\	
                && Proposed &1 &1 & 1 & 1 \\ 	 \cline{2-7}
				&\multirow{3}{*}{\textbf{Case II}}& SY2010 & 1 & 1 &  &    \\ 
				&&LC2012 & 1&1 &1 &1  \\ 
				&&CLX2013 & 0.947 & 1 & 1 & 1 \\
                && HC2018 & 1& 1& 1 &1  \\	
                && Proposed & 1&1 & 1 & 1   \\	 \cline{2-7}
				&\multirow{3}{*}{\textbf{Case III}}& SY2010 &  0& 0& 0.013 & 0.023   \\ 
				&&LC2012 & 0.218& 0.286 & 0.45 & 0.463  \\ 
				&&CLX2013 &0.067 & 0.057& 1 & 1 \\
                && HC2018 &1 &1 & 1 & 1 \\	
                && Proposed & 1& 1 & 1 &1 \\ 	 \hline
			\end{tabular}
			\caption{Comparison of simulated type I error and power for Gaussian samples. Here we choose the type I error $\alpha=0.05$ and consider the setups in Section \ref{sec_simulationsetup} for four different combinations of $(n_1,n_2)$ with $p=6,000.$ For Case II, we choose $\theta=0.5$ for the alternative and for Case III, we choose $\varepsilon=1$ for the alternative. In our $\mathtt{R}$ package $\texttt{UHDtst}$, our proposed method can be implemented using the function $\texttt{TwoSampleTest}$, LC2012 can be implemented using the function \texttt{LC2012}, CLX2013 can be implemented using the function \texttt{CLX2013}, SY2010 can be implemented using the function \texttt{SY2010} and HC2018 can be implemented using the function \texttt{HC2018}. We report the results based on 1,000 repetitions. } 		\label{caseGaussian}
		\end{table}

\begin{table}[htp]
			\fontsize{9pt}{7pt}\selectfont
			\centering
			\renewcommand{\arraystretch}{1.7}
					\hspace*{-10pt}
			\begin{tabular}{c|cc|cccc}
				\hline
%				\multicolumn{3}{c|}{}  \\
			%	\cline{4-11}
				\multicolumn{1}{c}{}& & \textbf{Methods/Setting}& $(100,100)$ & $(100,150)$ & (100,800) & (100,1000) \\
				\hline
				\multirow{9}{*}{\rotatebox[origin=c]{90}{\textbf{Empirical Size}}}&\multirow{3}{*}{\textbf{Case I}}& SY2010 & 0 & 0 & 0.01  & 0.01   \\ 
				&&LC2012 & 0.048 & 0.05 & 0.043& 0.046 \\ 
				&&CLX2013 & 0.176& 0.158 & 0.457  & 0.513 \\
                &&HC2018 & 0.01 & 0.014& 0.003  & 0.007  \\	
                &&Proposed & 0.049 & 0.051 & 0.048  & 0.047  \\ 	
				 \cline{2-7}
				&\multirow{3}{*}{\textbf{Case II}}& SY2010 & 0.051&0.033 &0.117  & 0.077   \\ 
				&&LC2012 & 0.046& 0.059 & 0.066 & 0.05 \\ 
				&&CLX2013 & 0.321& 0.306 & 0.95 & 0.997 \\
                && HC2018 & 0.011 & 0.013& 0.013 & 0.004 \\	
                && Proposed & 0.051& 0.052 & 0.048  & 0.048  \\ 	 \cline{2-7}
			&	\multirow{3}{*}{\textbf{Case III}}& SY2010 & 0 & 0 & 0.02 & 0.02    \\ 
				&&LC2012 & 0.044& 0.051&0.047 &0.06  \\ 
				&&CLX2013 &0.037 &0.035 & 0.2 & 0.233 \\
                && HC2018 &0.019 & 0.01& 0.006 & 0.008 \\	
                && Proposed & 0.051& 0.051& 0.047 & 0.048 \\ 				
				\hline
				\multirow{9}{*}{\rotatebox[origin=c]{90}{\textbf{Empirical Power}}}&\multirow{3}{*}{\textbf{Case I}}& SY2010 & 0 & 0 & 0.83  & 0.843   \\ 
				&&LC2012 & 1& 1& 1&1  \\ 
				&&CLX2013 & 1& 1& 1 & 1 \\
                && HC2018 & 1&1 &1  &1  \\	
                && Proposed &1 &1&  1& 1 \\ 	 \cline{2-7}
				&\multirow{3}{*}{\textbf{Case II}}& SY2010 & 1& 1& 1 & 1   \\ 
				&&LC2012 &1 &1 &1 &1  \\ 
				&&CLX2013 & 1&1 & 1 &1  \\
                && HC2018 &1 &1 & 1 & 1 \\	
                && Proposed & 1&1 & 0.996 &1    \\	 \cline{2-7}
				&\multirow{3}{*}{\textbf{Case III}}& SY2010 & 0 & 0 & 0 & 0.01    \\ 
				&&LC2012 &0.237 & 0.302 &0.413 & 0.39  \\ 
				&&CLX2013 & 0.038 & 0.051& 1 & 1 \\
                && HC2018 &1 &1 & 1 & 1 \\	
                && Proposed &1 & 1& 1 &1 \\ 	 \hline
			\end{tabular}
			\caption{Comparison of simulated type I error and power for two-point samples. Here we choose the type I error $\alpha=0.05$ and consider the setups in Section \ref{sec_simulationsetup} for four different combinations of $(n_1,n_2)$ with $p=6,000.$ For Case II, we choose $\theta=0.5$ for the alternative and for Case III, we choose $\varepsilon=1$ for the alternative. In our $\mathtt{R}$ package $\texttt{UHDtst}$, our proposed method can be implemented using the function $\texttt{TwoSampleTest}$, LC2012 can be implemented using the function \texttt{LC2012}, CLX2013 can be implemented using the function \texttt{CLX2013}, SY2010 can be implemented using the function \texttt{SY2010} and HC2018 can be implemented using the function \texttt{HC2018}. We report the results based on 1,000 repetitions. } 		\label{caseTwopoint}
		\end{table}

			\begin{figure}[htp]
		\centering
		\label{power}
		\includegraphics[width=1\textwidth]{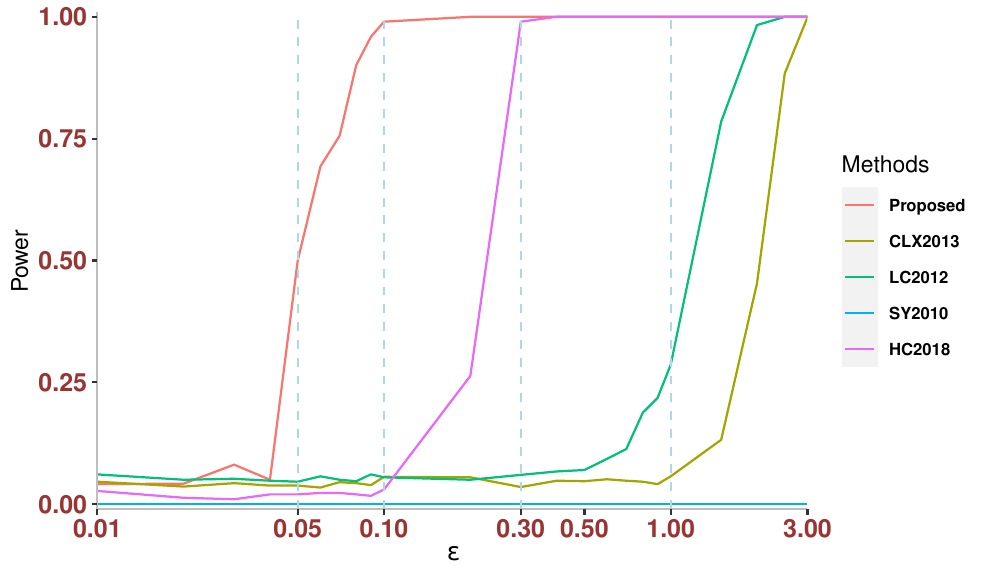} 
		% Adjust the width as needed
	\caption{Comparison of power of different methods under the alternative (\ref{eq_alternativeIII}) of Case III. Here we use the Gaussian samples and choose $n_1=100, n_2=150$ and $p=6,000.$ We report our methods based on $1,000$ repetitions. }	\label{fig_power}
	\end{figure}
%		
		
%		In our region, the only applicable methods are \cite{cai2013two, LC}. But need to add more. Some useful R codes are here \url{https://rdrr.io/cran/HDtest/man/testCov.html}
%		   
%\vspace*{3pt}		
%		
%		 
%		
%		{\color{blue}[need to add some simulations that their methods also work focus on comparison with other methods which work in the ultra-high dimensional setting.]}

		\subsection{Real data analysis}	
In this section, we consider the analysis of two gene expression data sets using our proposed method and compare it with the methods developed in \cite{cai2013two,he2018high,LC,srivastava2010testing}. The first data set is the clinical prostate cancer data set \cite{singh2002gene} \footnote{The data set can be downloaded from \url{https://www.ncbi.nlm.nih.gov/geo/query/acc.cgi?acc=GSE68907}} and the second one is the adult T-cell acute lymphocytic leukemia (ALL) data set \cite{chiaretti2004gene} \footnote{The data set can be loaded from the $\texttt{R}$ package $\texttt{ALL};$ see \url{https://bioconductor.org/packages/release/data/experiment/html/ALL.html}}. We will see from the analysis below that while some of these methods (including ours) work for the first data set, only our proposed method works for the second data set. 
		
		\subsubsection{Prostate cancer data}
The prostate cancer data set \cite{singh2002gene} focuses on the gene expression patterns associated with clinical behaviors of prostate cancer. This study employed microarray expression analysis to discern the global biological variations that might be linked to the common pathological characteristics of prostate cancer.

		The data set categorizes observations into distinct groups. More specifically, it has 12,600 columns  of gene expressions and comprises samples from two groups: a normal group with 50 samples and a tumor group with 52 samples. We point out that this dataset has been also used for analysis in \cite{cai2013two}. However, to avoid some computational issue, they only select the top 5,000 columns (genes) with the largest $t$-values in the sense of group means. In what follows, we will conduct our analysis on both this subsample with $5,000$ genes and all the samples with $12,600$ genes.  
		
	We conduct the two sample covariance tests both within groups and between groups. More concretely, for within group test, we consider the normal group and divide the 50 samples into two subgroups with sample sizes $n_1=30, n_2=20.$ For between group test, we use all 50 samples for normal group and all $52$ samples for tumor group, i.e., $n_1=50, n_2=52.$ The results are summarized in Table \ref{prostate} and we can make the following conclusions.  
	
	First, all of our proposed method, LC2012 and HC2018 will be able to accept the null hypothesis for the within group test and reject the null hypothesis for the between group test for both data sets with different numbers of genes. Second, SY2010 is able to accept the null hypothesis for the between group test but has no power to reject the null hypothesis. Third, for the method CLX2013, even though it works for the data set with 5,000 selected genes as in the work \cite{cai2013two}, it fails for the whole data set with all the genes.

%		For this dataset, we first standardize the whole data by columns and then we divide it into two groups.  In our data analysis, we will compare both 5000 dataset and the full dataset. We compare our methods with \cite{cai2013two}, \cite{LC} and \cite{srivastava2010testing}.

%		\paragraph{Within group comparison}
%		Here we only consider normal group. We random split the 50 samples into two sets with 25 samples. We set $n=15$ and the tuning parameter c to be 0.5 by Algorithm \ref{algo2}. We can see that all four methods accept the null hypothesis and our p value is the largest in the full data. 
		
%				\paragraph{Between groups comparison}
%		We have 50 and 52 samples for each group, and we here set $n=30$. Using Algorithm \ref{algo2}, we set $\eta_0^i =0.5*std(\gamma_j^i )$. We can see from Table \ref{prostate} that proposed methods can distinguish the difference between different groups either under 0.05 or 0.01 significance level. For SY2010, it fails to reject both and for CLX2013, even if it can reject the null in the 5000 selected columns, its performance becomes worse in the full data. This result shows the better performance of propose method.

		\begin{table}[htp]
			\centering
			\label{prostate}
			\renewcommand{\arraystretch}{1.2}
			\begin{tabular}{|p{5cm}|c|c|c|c|}
				\hline
				\multirow{2}{5cm}{\textbf{Methods}} & \multicolumn{2}{c|}{\textbf{Within Group}} & \multicolumn{2}{c|}{\textbf{Between Groups}}\\
				% \hline
				% \textbf{Inactive Modes} & \textbf{Description}\\
				\cline{2-5}
				& 5,000 genes & 12,600 genes & 5,000 genes & 12,600 genes\\
				%\hhline{~--}
				\hline
				SY2010 &  Accept & Accept & Accept & Accept\\ \hline
				LC2012 & Accept & Accept & Reject & Reject \\ \hline
				CLX2013 & Accept & Reject & Reject & Reject \\ \hline
				HC2018 & Accept & Accept & Reject  &  Reject \\ \hline
			    Proposed & Accept & Accept & Reject & Reject \\ \hline
			\end{tabular}
			\caption{Comparison of results for the prostate cancer data. Here 5,000 genes only uses $p=5,000$ genes as in \cite{cai2013two} and 126,000 genes contains all genes. }\label{prostate}
		\end{table}

		%\subsubsection{30-FTR}
		%$n=30$, not random, first 30 samples from both groups, $E=\lambda^{(1)}_2$, $\eta=std(eigvals1)$
		
		\subsubsection{Acute Lymphoblastic Leukemia data}
The second data set \cite{chiaretti2004gene} contains gene expression of adult T-cell acute lymphocytic leukemia (ALL) of patients with different biological indices. This study focuses on the relation between overall gene expressions and molecular biology types and helps to reveal the mechanism between different gene expressions of ALL and their corresponding responses to therapy and survival.
		
	 The data set contains 128 patients and their genes with length of 12,625. There are six types of molecular biology in total. Here we only select the two groups with largest numbers of patients, NEG with size $ 74$ and BCR/ABL with size $37$. We conduct our study on the whole gene sequence that $p=12,625$. 
		
		We conduct the two sample covariance tests both within groups and between groups. More concretely, for within group test, we consider the NEG group and divide the 74 samples into two subgroups with sample sizes $n_1=30,n_2=44.$ For between group test, we use all 74 samples from the NEG group and all $37$ samples from BCR/ABL group, i.e., $n_1=74, n_2=37.$ The results are summarized in Table \ref{all}.

We can see that for this data set, only our proposed method works while SY2010, LC2012 and HC2018 fail to reject the null hypothesis for between group test and CLX2013 rejects the null hypothesis for within group test.

		\begin{table}[htp]
			\centering
			\label{all}
			\renewcommand{\arraystretch}{1.2}
			\begin{tabular}{|c|c|c|}
				\hline
				\textbf{Methods} & \textbf{Within Group} & \textbf{Between Groups}\\
				\hline
				SY2010 &  Accept & Accept \\ \hline
				LC2012 & Accept & Accept \\ \hline
				CLX2013 & Reject & Reject \\ \hline
				HC2018 & Accept & Accept  \\ \hline
				Proposed & Accept & Reject \\ \hline
			\end{tabular}
			\caption{Comparison of results for the ALL data.}		\label{all}
		\end{table}

\section{Discussions}
		
In this paper, we consider the test of  equality of two population covariance matrices (c.f. (\ref{nullhypothesis})) of ultra-high dimensional (c.f. (\ref{eq_dimension})) random vectors. We propose a novel and adaptive test procedure which (i). does not require specific assumption (e.g., comparable or balancing, etc.) on the sizes of two samples; (ii). does not need quantitative or structural assumptions of the population covariance matrices; (iii). does not need the parametric distributions or the detailed knowledge of the moments of the two populations.

Our approach is summarized in Algorithm \ref{alg:bootstrapping}. It consists of three important components. The first one is a data splitting procedure in Algorithm \ref{alg0}. The second one is the construction of a sequence of statistics based on a subset of eigenvalues of some sample covariance matrices. The selection of the subset of the eigenvalues relies on a location parameter and a bandwidth parameter which can be chosen automatically using (\ref{eq_median}) and Algorithm \ref{algo2}. The third one is the calculation of a summary statistic (\ref{eq_decisionrationdefinition}) and a threshold $\delta$ which is automatically generated from a calibration procedure in Algorithm \ref{algo3}.

The proposed methodology is highly inspired and justified by our theoretical development in random matrix theory.  We establish the asymptotic distributions of the statistics used in our method and conduct the power analysis. We justify that our method is powerful under very weak alternatives. We also conduct extensive numerical simulations and show that our method  significantly outperforms the existing ones developed in \cite{cai2013two,he2018high,LC,srivastava2010testing}, both in terms of size and power. Analysis of two real data sets is also carried out to demonstrate the usefulness and superior performance of our proposed methodology. An $\texttt{R}$ package $\texttt{UHDtst}$ is developed for easy implementation of our proposed methodology.

Several further works can be considered following the spirit of the current paper. First, besides the two sample covariance matrix test, people are also interested in high dimensional two sample mean test under various settings, for example see  \cite{10.1214/09-AOS716,chen2019two,10.1214/19-AOS1848}.   It is important to propose an adaptive, accurate and powerful test for the two sample means under the ultra-high dimensional setup (\ref{eq_dimension}). Second, in the current paper, we assume that the eigenvalues of $\Sigma_1$ and $\Sigma_2$ are bounded from above and below away from zero which is a realistic assumption in many applications. However, in other applications where a factor model is more beneficial, people may consider the spiked covariance matrix models where a few larger or divergent eigenvalues are added \cite{fan2022estimating,ke2023estimation,zhang2023factor}. It will be interesting to generalize our results and methods to the spiked model. Finally, since our algorithm involves a multiple data splitting step, it is worth asking whether we can  implement our Algorithm \ref{alg:bootstrapping} in a parallel     or distributed fashion \cite{dobriban2019deterministic,10.1214/20-AOS1984}.

%	connection and comparison of existing methods

%Our methodology has been rigorously tested against some existing methods \cite{cai2013two, LC,he2018high,srivastava2010testing}, using both simulated data based on scenarios from \cite{cai2013two, chen2010tests} and an original setting we develop, as well as real geno-datasets \cite{singh2002gene, chiaretti2004gene}. Overall, our approach consistently outperforms these existing methods \cite{cai2013two, LC,he2018high,srivastava2010testing}. Notably, the method from \cite{srivastava2010testing} struggles in most settings. The performance of \cite{cai2013two}'s method significantly declines when $n_1$ and $n_2$ were not comparable or when the data entries are non-Gaussian, with the size potentially reaching 1. While \cite{chen2010tests}'s method is accurate in most cases, however, its size does not match the precision of our proposed method in our specific setting.  {\color{blue}  Conclusion for \cite{he2018high} here} These results underscore the theoretical robustness and practical effectiveness of our approach in the ultra-high-dimensional regime.

%future works

%In the present paper {\color{red}: what we did, connections with literature/existing methods, some future works}

		\bigskip
		\begin{center}
			{\large\bf SUPPLEMENTARY MATERIAL}
		\end{center}
		
%		\begin{description}
%\item		
In this supplement, we provide the details of the technical proof and the automated procedures for selecting the tuning parameters. 
			
%			\item[Title:] Brief description. (file type)
			
			%\item
%		\end{description}
		\setcounter{figure}{0} 
		\setcounter{table}{0}
		\setcounter{section}{0}  
		\counterwithin{table}{section}
		\counterwithin{figure}{section}
\appendix		
		
		\section{Technical proof}\label{sec_techinicalproof}
In this section, we provide the technical proofs 	for Theorems \ref{thm_onesample} and \ref{thm_poweranalysis}. For Theorem \ref{thm_onesample}, due to similarity, we only focus on $\mathbb{T}_x$. 
We assume that $\Sigma_1$ admits the following spectral decomposition 
		$$
		\Sigma_1=O_1^{\top} \Lambda_1 O_1, \quad \Lambda_1=\operatorname{diag}\left(\sigma_1^{(1)}, \ldots, \sigma_p^{(1)}\right),
		$$
and $O_1^\top$ contains the eigenvectors of $\Sigma_1.$	 It is clear that $\mathcal{Q}_x$ in (\ref{eq_samplecovaraincematrixdefinition}) shares the same eigenvalues with 
		\begin{equation*}
			\mathcal{Q}_1=\Lambda_1^{1/2} O_1 \mathsf{X}\mathsf{X}^\top O_1^{\top} \Lambda_1^{1/2}, \ \ \mathsf{X}=(pn)^{-1/4} X.  
		\end{equation*} 
In what follows, for simplicity of presentation, whenever there is no danger of confusion,  we suppress the subscript and superscript related to $X$ and simply denote
\begin{equation}\label{eq_simplifiednotation}
 m(z)\equiv m_1(z),\quad \varrho \equiv \varrho_1,
\end{equation}		
\begin{equation}\label{eq_othernotations}
	\Sigma\equiv \Sigma_1,\quad O\equiv O_1,\quad \Lambda\equiv \Lambda_1=\operatorname{diag}\left(\sigma_1, \ldots, \sigma_p\right),\quad
	\mathcal{Q}\equiv\mathcal{Q}_1, \quad Q\equiv Q_x.
		\end{equation}

Furthermore, for $z \in \mathbb{C}_+,$ we define the resolvents as follows
\begin{equation}\label{eq_resolventdefinition}
	\begin{aligned}
		G(z)=&( \mathsf{X}^{\top} \Sigma \mathsf{X}-z)^{-1}, \quad
		\mathcal{G}(z)=&(\Lambda^{1/2} O\mathsf{X} \mathsf{X}^\top O^{\top} \Lambda^{1/2}-z)^{-1}.
	\end{aligned}
\end{equation}

		\subsection{Preliminary results}
In this section, we introduce some preliminary results. First, we provide an more explicit expression for the asymptotic densities associated with $\varrho_1$ from Lemma \ref{lem_density}. Recall the conventions in (\ref{eq_emperical}).  

\begin{lemma}\label{lem_densityestimatation}
Suppose the assumptions of Lemma \ref{lem_density} hold. Then we have that when $n$ is sufficiently large
\begin{equation}\label{eq_semicircleapproximation}
\mathrm{d} \varrho_1(x)= \left( \frac{\sqrt{4 \left(\fm_2(\Sigma_1)+\mathrm{O}(\phi^{-1/2}) \right)-(x-\fm_1(\Sigma_1)\phi^{1/2})^2}}{2 \pi \left(\fm_2(\Sigma_1)+\mathrm{O}(\phi^{-1/2})\right)} \right) \mathrm{d} x. 
\end{equation}
\end{lemma}
\begin{proof}
Using the conventions in (\ref{eq_emperical}) and Lemma \ref{lem_density}, we can rewrite that  
\begin{equation}\label{eq_expansionform}
\frac{1}{m_1}=-z+\int \frac{\phi}{\phi^{1/2} x^{-1}+m_1} \pi_1(\mathrm{d} x). 
\end{equation}
By Lemma 5.1 of \cite{DingWang2023}, we have that $m_1 \asymp 1.$ Therefore, for $x \in \operatorname{supp}(\pi_1)$ and sufficiently large $n,$ we have that $x m_1 \phi^{-1/2}=\mathrm{o}(1).$ This yields that 
\begin{equation*}
\frac{1}{1+\phi^{-1/2}x m_1}=1-\phi^{-1/2} x m_1+\mathrm{O}(\phi^{-1}).  
\end{equation*} 
Insert the above equation into (\ref{eq_expansionform}), we obtain that 
\begin{equation*}
\frac{1}{m_1}=-z+\phi^{1/2} \int x(1-\phi^{-1/2}x m_1) \pi_1(\mathrm{d} x)+\mathrm{O}(\phi^{-1/2}), 
\end{equation*}
which together with the convention (\ref{eq_emperical}) yields that 
\begin{equation*}
\frac{1}{m_1}=-z+\phi^{1/2} \fm_1(\Sigma_1)-\left(\fm_2(\Sigma_1) m_1+\mathrm{O}(\phi^{-1/2}) \right).
\end{equation*}
Again using the fact that $m_1 \asymp 1,$ we can further rewrite the above equation into a quadratic form that 
\begin{equation}\label{eq_stiletjestransformquadratic}
\left(\mathfrak{m}_2(\Sigma_1)+\mathrm{O}(\phi^{-1/2}) \right) m_1^2+(z-\phi^{1/2} \mathfrak{m}_1(\Sigma_1))m_1+1=0.
\end{equation}
By the definition of the Stieltjes transform in (\ref{eq_definitionStieltjestransform}), we  have that $m_1(z)$ maps $\mathbb{C}_+$ to $\mathbb{C}+_.$ Solving the quadratic equation (\ref{eq_stiletjestransformquadratic}) and finding the solution that maps $\mathbb{C}_+$ to $\mathbb{C}_+,$ we find that for $\mathrm{Re} \ z \in [\phi^{1/2} \mathfrak{m}_1(\Sigma_1)-2 \sqrt{\mathfrak{m}_2(\Sigma_1)+\mathrm{O}(\phi^{-1/2})}, \phi^{1/2} \mathfrak{m}_1(\Sigma_1)+2 \sqrt{\mathfrak{m}_2(\Sigma_1)+\mathrm{O}(\phi^{-1/2})}],$ 
\begin{equation}\label{eq_m1form}
m_1(z)=\frac{-(z-\phi^{1/2} \mathfrak{m}_1(\Sigma_1))+\mathrm{i} \sqrt{4 (\mathfrak{m}_2(\Sigma_1)+\mathrm{O}(\phi^{-1/2}))-(z-\phi^{1/2} \mathfrak{m}_1(\Sigma_1))^2 }}{2 (\mathfrak{m}_2(\Sigma_1)+\mathrm{O}(\phi^{-1/2}))}.
\end{equation} 
Recall the inverse formula that (for example see Section B.2 of \cite{MR2567175})
\begin{equation}\label{eq_inverseformula}
\mathrm{d} \varrho_1(x)=\frac{1}{\pi} \lim_{\eta \downarrow 0} m_1(x+\mathrm{i} \eta). 
\end{equation}
Together with (\ref{eq_m1form}), we can complete the proof.  
\end{proof}

Then we summarize and prove the local laws which are the key ingredients for our technical proof. For some small fixed constant $0<\tau<1,$ denote the set of the spectral parameters as
\begin{equation}\label{eq_parameterspectral}
\mathbf{S} \equiv \mathbf{S}(\tau):= \left\{ z=E+\mathrm{i}\eta \in \mathbb{C}: |E-\phi^{1/2} \mathfrak{m}_1| \leq \tau^{-1}, \ n^{-1+\tau} \leq \eta \leq \tau^{-1} \right\}, 
\end{equation}
where we recall $\mathfrak{m}_1=\int x \pi(\mathrm{d} x).$ For $z \in \mathbf{S},$ $m(z)$ in (\ref{eq_simplifiednotation}) and $\Lambda$ in (\ref{eq_othernotations}), we further denote  
\begin{equation}\label{eq_Pi(z)}
\Pi(z) \equiv \Pi(z, \Sigma):=-z^{-1}\left(I_p+\phi^{-1 / 2} m(z) \Lambda\right)^{-1}, 
\end{equation}
and 
\begin{equation}\label{eq_psi(z)}
\Psi(z):=\sqrt{\frac{\operatorname{Im} m(z)}{n \eta}}+\frac{1}{n \eta}.
\end{equation}  

Then the local laws are summarized in the following theorem. We will need the following notion of stochastic domination which is commonly used in modern random matrix theory \cite{erdHos2017dynamical}. 
\begin{definition}[Stochastic domination]
(i) Let
\[A=\left(A^{(n)}(u):n\in\mathbb{N}, u\in U^{(n)}\right),\hskip 10pt B=\left(B^{(n)}(u):n\in\mathbb{N}, u\in U^{(n)}\right),\]
be two families of nonnegative random variables, where $U^{(n)}$ is a possibly $n$-dependent parameter set. We say $A$ is stochastically dominated by $B$, uniformly in $u$, if for any fixed (small) $\epsilon>0$ and (large) $D>0$, 
\[\sup_{u\in U^{(n)}}\mathbb{P}\left(A^{(n)}(u)>n^\epsilon B^{(n)}(u)\right)\le n^{-D},\]
for large enough $n \ge n_0(\epsilon, D)$, and we shall use the notation $A\prec B$. Throughout this paper, the stochastic domination will always be uniform in all parameters that are not explicitly fixed, such as the matrix indices and the spectral parameter $z$.  Throughout the paper, even for negative or complex random variables, we also write $A \prec B$ or $A=\mathrm{O}_\prec(B)$ for $|A| \prec B$.

\noindent (ii) We say an event $\Xi$ holds with high probability if for any constant $D>0$, $\mathbb P(\Xi)\ge 1- n^{-D}$ for large enough $n$.
\end{definition}  

		\begin{theorem}\label{thm_locallaw}
			Suppose Assumption \ref{assum_basic} holds. For all $z \in \mathbf{S}$ uniformly and any deterministic vectors $\mathbf{u}_k \in \mathbb{R}^p, \mathbf{v}_k \in \mathbb{R}^{n}, k=1,2$, we have that:
\begin{enumerate}			
			\item[(1).] For the resolvents of $Q$ and $\mathcal{Q}$, we have that
			$$
			\mathbf{u}_1^* \mathcal{G}(z) \mathbf{u}_2=\mathbf{u}_1^* \Pi(z) \mathbf{u}_2+\mathrm{O}_{\prec}\left(\phi^{-1} \Psi(z)\right),
			$$
			where $\Pi(z)$ is defined in (\ref{eq_Pi(z)}) and $\Psi(z)$ is defined in (\ref{eq_psi(z)}). Moreover, we have that
			$$
			\mathbf{v}_1^* G(z) \mathbf{v}_2=m(z) \mathbf{v}_1^* \mathbf{v}_2+\mathrm{O}_{\prec}(\Psi(z)) .
			$$
			\item[(2).] Denote the normalized resolvents as
			$$
			m_{1 n}(z)=\frac{1}{p} \sum_{i=1}^p\mathcal{G}_{ii}(z), m_{2 n}(z)=\frac{1}{n} \sum_{i=1}^n G_{ii}(z).
			$$
			We have that
			$$
			m_{1 n}(z)=-\frac{1}{z p} \sum_{i=1}^p \frac{1}{1+\sigma_i \phi^{-1 / 2} m(z)}+\mathrm{O}_{\prec}\left((p \eta)^{-1}\right),
			$$
			and
			$$
			m_{2 n}(z)=m(z)+\mathrm{O}_{\prec}\left((n \eta)^{-1}\right) .
			$$
			\item[(3).] Denote $Y=\Lambda^{1/2}OX.$ For $1 \leq i \leq p$ and $1 \leq j \leq n,$ we have that 
			\begin{equation*}
			(GY^*)_{ji}=\mathrm{O}_{\prec}(\phi^{-1/4} \Psi(z)).
			\end{equation*}
%			(3). Recall the conventions for the indices in Definition 5.3. For $G(z)$ defined in (5.10), we have that uniformly in $\mu \in \mathcal{I}_2$ and $i \in \mathcal{I}_1$
%			$$
%			G_{i \mu}(z)=\mathrm{O}_{\prec}\left(\phi^{-1 / 4} \Psi(z)\right) .
%			$$
\end{enumerate}
		\end{theorem}
		\begin{proof}
The results have been proved in Theorem 5.5 of \cite{DingWang2023} when $\Sigma$ is diagonal in the sense that $O=I$ as in (\ref{eq_othernotations}). For general $O$ and the associated resolvents in (\ref{eq_resolventdefinition}), we can follow the discussions of the proofs of Theorem 3.6 of \cite{FYEJP}. The arguments in \cite{FYEJP} consist of two steps. In the first step, we need to establish the results for Gaussian $X.$ Clearly, when $X$ is Gaussian, it suffices to assume that $\Sigma$ is diagonal so that Theorem 5.5 of \cite{DingWang2023} implies Theorem \ref{thm_locallaw}. In the second step, we need to conduct a comparison argument 	as in Sections 5 and 6 of \cite{FYEJP}. In fact, the forms of the resolvents in (\ref{eq_resolventdefinition}) are in the same fashion as those in equations (2.6) and (2.7) of \cite{FYEJP} so that we can follow lines of the arguments in Sections 5 and 6 of \cite{FYEJP} to complete the proof. We omit further details. 	
		\end{proof}
An important consequence of Theorem \ref{thm_locallaw} is the result concerning the rigidity of the eigenvalues. For $1 \leq i \leq n,$ denote the sequence of classical locations of $\varrho$ as $\{\omega_i\}$ that
\begin{equation}
\int_{\omega_i}^{\infty} \mathrm{d} \varrho(x)=\frac{i}{n}.
\end{equation} 
Denote the nonzero eigenvalues of $Q$ in the decreasing order as $\lambda_1 \geq \lambda_2 \geq \cdots \lambda_n>0.$
\begin{lemma}\label{lem_rigidity} Suppose Assumption \ref{assum_basic} holds. Then we have that 
\begin{equation*}
|\lambda_i-\omega_i| \prec \left( \min \left\{ i, n+1-i \right\}  \right)^{-1/3} n^{-2/3}. 
\end{equation*}
\end{lemma}
\begin{proof}
The proof follows from a routine application of the Helffer-Sj{\"o}strand functional calculus; see Chapter 11 of \cite{erdHos2017dynamical} or \cite[Section 4.3]{BEYY}. 
\end{proof}		
\begin{remark}\label{rem_rigidity}
Two remarks are in order. First, when the null hypothesis (\ref{nullhypothesis}) holds, using our splitting scheme Algorithm \ref{alg0}, it is easy to see from Lemma \ref{lem_density} that the asymptotic density associated with the split data sets $\mathcal{X}^s, \ \mathcal{Y}^s, \ \mathcal{Z}^s$ are identical, denoted as $\varrho$. Consequently, for $\gamma$ defined in (\ref{eq_median}), we see from Lemma \ref{lem_rigidity} that 
\begin{equation}\label{eq_controlone}
\gamma=\omega+\mathrm{O}_{\mathbb{P}}(n^{-1}), 
\end{equation}
where 
\begin{equation*}
\omega:=
\begin{cases}
\omega_{(n+1)/2}, & \ \text{if} \ n \ \text{is odd} \\
\frac{\omega_{n/2}+\omega_{n/2+1}}{2}, & \ \text{if} \ n \ \text{is even}
\end{cases}.
\end{equation*}
Then if we set $\gamma_{\pm}$ as the endpoints of the support of $\varrho,$ it is clear that  with high probability, 
\begin{equation*}
\gamma_-<\gamma<\gamma_+.
\end{equation*}

Second, according to Lemma \ref{lem_densityestimatation} and the conventions in (\ref{eq_othernotations}), we have that 
\begin{equation*}
\gamma_{\pm}=\mathfrak{m}_1(\Sigma) \phi^{1/2}\pm \mathfrak{m}_2(\Sigma)+\mathrm{O}(\phi^{-1/2}).
\end{equation*}
Moreover, since (\ref{eq_semicircleapproximation}) is of a semicircle form, we have that
\begin{equation}\label{eq_approximate}
\omega=\mathfrak{m}_1(\Sigma) \phi^{1/2}+\mathrm{O}(\phi^{-1/2}). 
\end{equation}
\end{remark}

\subsection{Proof of Theorem \ref{thm_poweranalysis}}	\label{sec_32proofdetails}	
Without loss of generality, in what follows, we assume that $\mathcal{X}$ has more samples than $\mathcal{Y}$ in the sense that $n_1 \geq n_2.$ Then according to the splitting procedure Algorithm \ref{alg0} and Lemma \ref{lem_rigidity} (or (\ref{eq_controlone}) and (\ref{eq_approximate})), we have that
\begin{equation}\label{eq_convergence}
\gamma=\mathfrak{m}_1(\Sigma_1) \phi^{1/2}+\mathrm{O}_{\mathbb{P}}(n^{-1}+\phi^{-1/2}). 
\end{equation}
Since $(\mathcal{X}^s, \mathcal{Y}^s, \mathcal{Z}^s)$ is an $\epsilon$-efficient data splitting, by Definition \ref{defn_inefficient} and Lemma \ref{lem_rigidity}, we find that the assumptions of Theorem \ref{thm_onesample} hold with high probability. Consequently, we have that 
\begin{equation}\label{eq_Talternativedistribution}
\mathbb{T}- \left( \mathsf{M}_x-\mathsf{M}_y \right) \Rightarrow \mathcal{N}(0, 2 \mathsf{v}). 
\end{equation}
Therefore, it suffices to study $\mathsf{M}_x-\mathsf{M}_y$ under $\mathbf{H}_a$ in (\ref{hypothesis_alternative}).  By the definitions in (\ref{mean}), we have that 
\begin{equation}\label{eq_decompositiononeoneone}
\mathsf{M}_x-\mathsf{M}_y= n \int_{\mathbb{R}} \frac{t-\gamma}{\eta_0} \mathcal{K} \left( \frac{t-\gamma}{\eta_0} \right)\left[ \mathrm{d} \varrho_1(t)-\mathrm{d} \varrho_2(t) \right].
\end{equation} 
According to Lemma \ref{lem_densityestimatation}, we see that 
\begin{align*}
\mathrm{d} \varrho_1(t)-\mathrm{d} \varrho_2(t)= \Biggl( & \frac{\sqrt{4  \left(\mathfrak{m}_2(\Sigma_1)  +\mathrm{O}(\phi^{-1/2})\right)-(t-\mathfrak{m}_1(\Sigma_1) \phi^{1/2})^2}}{2 \pi \left( \mathfrak{m}_2(\Sigma_1)+\mathrm{O}(\phi^{-1/2}) \right)} \\
&-\frac{\sqrt{4 \left(\mathfrak{m}_2(\Sigma_2)+\mathrm{O}(\phi^{-1/2}) \right)-(t-\mathfrak{m}_1(\Sigma_2) \phi^{1/2})^2}}{2 \pi \left(\mathfrak{m}_2(\Sigma_2)+\mathrm{O}(\phi^{-1/2}) \right)} \Biggr)\mathrm{d}t. 
\end{align*}
In what follows, for notional simplicity, for $i=1,2,$ we rewrite
$$\mathrm{d}\varrho_i(t) \equiv d_i(t)\mathrm{d}t:=\frac{2\sqrt{\mathfrak{a}_i^2-(t-\mathfrak{b}_i)^2   }   }{\pi\mathfrak{a}_i^2  }\mathrm{d} t,$$ where we used the short-hand notations that  $\mathfrak{a}_i:=2 \sqrt{\mathfrak{m}_2\left(\Sigma_i\right)+\mathrm{O}\left(\phi^{-1 / 2}\right)}, \mathfrak{b}_i:=\mathfrak{m}_1\left(\Sigma_i\right) \phi^{1 / 2}$. 
In view of the definition of $\mathcal{K}(x)$ in (\ref{eq_mathcalK}), we only need to calculate the integral on the interval $[\gamma-\mathrm{O}(\eta_0), \gamma+\mathrm{O}(\eta_0)]$ whose length is just $\mathrm{O}(\eta_0)$. 

Using the above notations, in view of (\ref{eq_decompositiononeoneone}), we denote that
\begin{align*}
	I:=\int_{\gamma-1.05\eta_0}^{\gamma+1.05\eta_0} \frac{t-\gamma}{\eta_0}\mathcal{K}\left(\frac{t-\gamma}{\eta_0}\right) (d_1(t)-d_2(t)) \mathrm{d}t.
\end{align*}
In what follows, we expand $d_i(t), i=1,2,$ around $\gamma$ to the $(2k)$th order, where $k$ is some finite positive integer to be determined in the end.  More specifically, since the splitting is $\epsilon$-efficient, we can rewrite $I$ as  
\begin{equation*}
I=\int_{\gamma-1.05\eta_0}^{\gamma+1.05\eta_0} \frac{t-\gamma}{\eta_0}\mathcal{K}\left(\frac{t-\gamma}{\eta_0}\right) \left(\sum_{j=0}^{2k}\frac{(t-\gamma)^j}{j!}(d_1^{(j)}(\gamma)-d_2^{(j)}(\gamma))+\mathrm{o}\left((t-\gamma)^{2k} \right)\right) \mathrm{d}t.
\end{equation*}
After a straightforward calculation using the definition of $\mathcal{K}(x)$ in (\ref{eq_mathcalK}), we can further obtain that  
\begin{align}\label{eq_maincontrolterm}	
I=\int_{0}^{1.05\eta_0} \frac{1}{\eta_0}\mathcal{K}\left(\frac{s}{\eta_0}\right) \left(\sum_{j=1}^{k}\frac{2(s)^{2j}}{(2j-1)!}(d_1^{(2j-1)}(\gamma)-d_2^{(2j-1)}(\gamma))+\mathrm{o}(s^{2k+1})\right) \mathrm{d}s.  
\end{align}

First, the error term can be controlled as $$\int_{0}^{1.05\eta_0} \frac{1}{\eta_0}\mathcal{K}\left(\frac{s}{\eta_0}\right)\mathrm{o}(s^{2k+1}) \mathrm{d}s=\mathrm{o}(\eta_0^{2k+1}).$$
Second, let $\tilde{s}:=s/\eta_0$ and we observe that for all $1 \leq j \leq k$
 \begin{align*}
	&\frac{2}{(2j-1)!\eta_0}\int_{0}^{1.05\eta_0}\mathcal{K}\left(\frac{s}{\eta_0}\right)s^{2j}\mathrm{d}s\\ 
&	=  	\frac{2}{(2j-1)!\eta_0}\left(\int_{0}^{\eta_0}+\int_{\eta_0}^{1.05\eta_0}\right)\mathcal{K}\left(\frac{s}{\eta_0}\right)s^{2j}\mathrm{d}s\\ 
&	=  	\frac{2}{(2j-1)!\eta_0}\int_{0}^{\eta_0}s^{2j}\mathrm{d}s+\frac{\eta_0^{2j}}{(2j-1)!}\int_{1}^{1.05}\mathcal{K}\left(\tilde{s}\right)\left(\tilde{s}\right)^{2j}\mathrm{d}\left(\tilde{s}\right)\\ 
&	=\left( \frac{2}{(2j-1)!(2j+1)}+\frac{\int_{1}^{1.05}\mathcal{K}\left(\tilde{s}\right)\left(\tilde{s}\right)^{2j}\mathrm{d}\left(\tilde{s}\right)}{(2j-1)!} \right)\eta_0^{2j}\\ 
&	= c(j)\eta_0^{2j},
\end{align*}
where   $c(j)$ is a positive constant only depending on  $j$, since $\int_{1}^{1.05}\mathcal{K}\left(\tilde{s}\right)\left(\tilde{s}\right)^{2j}\mathrm{d}\left(\tilde{s}\right)$ is a scale-free positive constant.

Third, we proceed to understand $d_1^{(j)}(\gamma)-d_2^{(j)}(\gamma).$ Due to similarity, we focus our discussion on $j=1.$ By a straightforward but tedious calculation, we have that  
\begin{align*}& \pi(d_1^{(1)}(\gamma)-d_2^{(1)}(\gamma))\\
	=&-\frac{2(\gamma-\mathfrak{b}_1)}{\mathfrak{a}_1^2\sqrt{\mathfrak{a}_1^2-(\gamma-\mathfrak{b}_1)^2}}+\frac{2(\gamma-\mathfrak{b}_2)}{\mathfrak{a}_2^2\sqrt{\mathfrak{a}_2^2-(\gamma-\mathfrak{b}_2)^2}}\\
	=&-2(\mathfrak{a}_2^2-\mathfrak{a}_1^2)\frac{(\gamma-\mathfrak{b}_1)\left(\mathfrak{a}_2^4+\mathfrak{a}_2^2\mathfrak{a}_1^2+\mathfrak{a}_1^4-(\gamma-\mathfrak{b}_1)^2(\mathfrak{a}_2^2+\mathfrak{a}_1^2)\right)}{\mathfrak{a}_1^2\sqrt{\mathfrak{a}_1^2-(\gamma-\mathfrak{b}_1)^2}\mathfrak{a}_2^2\sqrt{\mathfrak{a}_2^2-(\gamma-\mathfrak{b}_2)^2}\left(\mathfrak{a}_1^2\sqrt{\mathfrak{a}_1^2-(\gamma-\mathfrak{b}_1)^2}+\mathfrak{a}_2^2\sqrt{\mathfrak{a}_2^2-(\gamma-\mathfrak{b}_2)^2}\right)},\\ &-2(\mathfrak{b}_2-\mathfrak{b}_1)\left(\frac{\mathfrak{a}_2^4((\gamma-\mathfrak{b}_1)+(\gamma-\mathfrak{b}_2))(\gamma-\mathfrak{b}_1)}{\mathfrak{a}_1^2\sqrt{\mathfrak{a}_1^2-(\gamma-\mathfrak{b}_1)^2}\mathfrak{a}_2^2\sqrt{\mathfrak{a}_2^2-(\gamma-\mathfrak{b}_2)^2}\left(\mathfrak{a}_1^2\sqrt{\mathfrak{a}_1^2-(\gamma-\mathfrak{b}_1)^2}+\mathfrak{a}_2^2\sqrt{\mathfrak{a}_2^2-(\gamma-\mathfrak{b}_2)^2}\right)}\right.\\ 
	&\left.+\frac{1}{\mathfrak{a}_2^2\sqrt{\mathfrak{a}_2^2-(\gamma-\mathfrak{b}_2)^2}}\right).
\end{align*}
In view of (\ref{eq_convergence}), we can see that for some constant $\mathsf{c}=\mathrm{O}(n^{-1}+\phi^{-1/2}),$ with high probability
\begin{equation*}
d_1^{(1)}(\gamma)-d_2^{(1)}(\gamma) \asymp \mathsf{c}(\mathfrak{a}_2-\mathfrak{a}_1)+(\mathfrak{b}_2-\mathfrak{b}_1), 
\end{equation*}  
where we used the assumption that the splitting is $\epsilon$-efficient. We can repeat the discussions for $j \geq 1$ and conclude that 
\begin{equation*}
d_1^{(j)}(\gamma)-d_2^{(j)}(\gamma) \asymp \left[\mathsf{c}(\mathfrak{a}_2-\mathfrak{a}_1)+(\mathfrak{b}_2-\mathfrak{b}_1) \right](1+\mathrm{o}_{\mathbb{P}}(1)). 
\end{equation*} 

Inserting all the above controls back into (\ref{eq_maincontrolterm}), under the assumption of (\ref{eq_condition}), by choosing $k\ge -2\frac{\log n}{\log \eta_0}\asymp 1$ such that $n\eta_0^{2k+1}=\mathrm{o}(1),$ we find that with high probability 
\begin{equation*}
nI \asymp n\eta_0^2 (\mathsf{c} |\mathfrak{a}_1-\mathfrak{a}_2|+ |\mathfrak{b}_1-\mathfrak{b}_2|). 
\end{equation*}
In view of (\ref{eq_decisionrationdefinition}) and (\ref{eq_Talternativedistribution}), we can complete the proof of (\ref{eq_Tpower}) under the assumption of (\ref{eq_localalternative}).

		\subsection{Proof of Theorem \ref{thm_onesample}}\label{suppl_thmonesampleproof}
In this section, we prove Theorem \ref{thm_onesample}. Similar results have been partially proved in \cite{DingWang2023} for diagonal $\Sigma$ which is quite restricted for two sample tests. Due to similarity, we mainly focus on explaining how to remove the diagonal assumption. For self-completeness, in Section \ref{sec_strategies}, we summarize and modify the proof strategies used in \cite{DingWang2023} to accommodate to our setting which basically follow and generalize those used in \cite{10.1214/20-AIHP1086,yang2020linear}. Then we provide some details of the proof in  Section \ref{sec_31proofdetails}. 	
		
		\subsubsection{Proof of the main theorem}\label{sec_strategies}
For notional simplicity, with the mollifier in (\ref{eq_mathcalK}), we define
\begin{equation}\label{eq_fxdefinition}
f(x) \equiv f(x; \gamma, \eta_0)=\frac{x-\gamma}{\eta_0} \mathcal{K}\left(\frac{x-\gamma}{\eta_0} \right). 
\end{equation} 
For $\mathbb{T}_x$ in (\ref{eq_2sampleindividualstatdef}) and $\mathsf{M}_x$ in (\ref{mean}), we first decompose it as follows 
\begin{equation*}
\mathbb{T}_x:=\left[\Tr f(Q)- \mathbb{E}\Tr f(Q) \right]+\left[\mathbb{E}\Tr f(Q)-\mathsf{M}_x\right]:=\mathcal{T}_1+\mathcal{T}_2. 
\end{equation*}
To prove Theorem \ref{thm_onesample}, it suffices to show that 
\begin{equation}\label{eq_needproofequationone}
\mathcal{T}_1  \Rightarrow \mathcal{N}(0, \mathsf{v}),  
\end{equation}
and 
\begin{equation}\label{eq_needproofequationtwo}
\mathcal{T}_2=\mathrm{o}_{\prec}(1).  
\end{equation}

We first discuss how to prove  (\ref{eq_needproofequationone}). The proof relies on the fact that a Gaussian random variable $z \sim \mathcal{N}(0, \mathsf{v})$ can be completely characterized by its moments which follow the recursive relation that for $l \geq 2$
\begin{equation*}
\mathbb{E}z^{l}=(l-1) \mathsf{v} \mathbb{E}z^{l-2}.
\end{equation*}
Consequently, we aim to prove that 
\begin{equation}\label{eq_T1needstobeproved}
\mathbb{E} \mathcal{T}_1^l=(l-1) \mathsf{v} \mathbb{E} \mathcal{T}_1^{l-2}+\mathrm{o}_{\prec}(1). 
\end{equation}

In order to prove (\ref{eq_T1needstobeproved}), we need to connect $\mathcal{T}_1$ with the resolvents (\ref{eq_resolventdefinition}) so that Theorem \ref{thm_locallaw} can be applied. Such a connection can be built using a complex integral via Helffer-Sj{\" o}strand formula. To be more specific, by Lemma 5.8 and a discussion similar to (7.23) of \cite{DingWang2023}, we have that   
\begin{equation}\label{eq_T1integral}
\mathcal{T}_1=\frac{1}{\pi} \int_{\mathbb{C}} \frac{\partial}{\partial \bar{z}} \widetilde{f}(z) \left[ \Tr \mathcal{G}(z)-\mathbb{E} \Tr \mathcal{G}(z) \right] \mathrm{d}^2 z,
\end{equation}
where $\widetilde{f}(z)$ is the almost-analytic extension of $f.$ For notional simplicity, we define 
\begin{equation*}
\theta_f(z):=\frac{1}{\pi} \frac{\partial \widetilde{f}(z)}{\partial \bar{z}}.
\end{equation*}
In order to prove (\ref{eq_T1needstobeproved}), we need to calculate 
\begin{equation}\label{eq_expectationexpansion}
\mathbb{E} \mathcal{T}_1^l=\int \theta_f(z_1) \cdots \theta_f(z_l) \mathbb{E} \left[ \mathcal{Y}(z_1) \cdots \mathcal{Y}(z_l) \right] \mathrm{d}^2 z_1 \cdots \mathrm{d}^2 z_l,
\end{equation}
where we denote 
\begin{equation}	\label{eq_yzform}
	 \mathcal{Y}(z_i) = \operatorname{Tr} \mathcal{G}(z_i) - \mathbb{E} \operatorname{Tr} \mathcal{G}(z_i), \ z_i \in \mathbb{C}. 
	  \end{equation}
Therefore, it suffices to calculate $\mathbb{E} \left[ \mathcal{Y}(z_1) \cdots \mathcal{Y}(z_l) \right].$ We summarize the results in the following lemma whose discussion of the proof will be provided in Section \ref{sec_31proofdetails}. 
\begin{lemma}\label{lem_momentresults}
Suppose the assumptions of Theorem \ref{thm_onesample} hold.  Fix $l \in \mathbb{N}$ and $z_1, \cdots, z_l \in \mathbf{S}$ in (\ref{eq_parameterspectral}), we have that 
\begin{equation*}
\left(\prod_{i=1}^l \eta_i \right) \mathbb{E} \prod_{i=1}^l \mathcal{Y}(z_i)=
\begin{cases}
\left(\prod_{i=1}^l \eta_i \right) \sum \prod \omega(z_s, z_t)+\mathrm{o}_{\prec}(1),  \ & l \in 2\mathbb{N} \\
\mathrm{o}_{\prec}(1), \ & \text{otherwise}
\end{cases}.
\end{equation*}
Here $\sum \prod$ means summing over all distinct ways of partitioning of indices into pairs  and $\omega(z_s, z_t):=\Delta_1(z_s, z_t)+\Delta_2(z_s, z_t)$ with
\begin{equation*}
\Delta_1(z_s, z_t):=\kappa_4 \phi \frac{\partial^2}{\partial z_s \partial z_t} \left( \frac{1}{p} \sum_{i=1}^p  \frac{1}{(1+\phi^{-1/2} m(z_s) \sigma_i)(1+\phi^{-1/2} m(z_t) \sigma_i)}\right),
\end{equation*}
\begin{equation*}
\Delta_2(z_s, z_t):=2 \left( \frac{m'(z_s) m'(z_t)}{(m(z_s)-m(z_t))^2}-\frac{1}{(z_s-z_t)^2} \right),
\end{equation*}
where $\kappa_4$ is the 4th cumulant of $x_{ij}$ in the sense that  $\kappa_4:=(-\mathrm{i})^4 \frac{\mathrm{d}}{\mathrm{d} t} \log \mathbb{E} e^{\mathrm{i}h x_{ij}}|_{h=0}.$

 Moreover, we have the recursive relation that  
\begin{equation}\label{eq_recursive}
\left(\prod_{i=1}^l \eta_i \right) \mathbb{E} \prod_{i=1}^l \mathcal{Y}(z_i)=\left(\prod_{i=1}^l \eta_i \right) \sum_{s=2}^l \omega(z_1, z_s) \mathbb{E} \prod_{t \notin \{1,s\}} \mathcal{Y}(z_t)+\mathrm{o}_{\prec}(1).
\end{equation}
\end{lemma}

Now we can insert the results of Lemma \ref{lem_momentresults} into (\ref{eq_expectationexpansion}) to prove (\ref{eq_T1needstobeproved}). The proofs are almost identical to that Lemma 7.1 of \cite{DingWang2023}. Due to similarity, we only sketch the main ideas and refer the readers to \cite{DingWang2023} for more details. For $1 \leq i \leq l,$ we set $z_i=x_i+\mathrm{i} y_i.$ The first step is to split the regions of $\mathbb{C}^l$ into two self-disjoint regions that $\mathbb{C}^l=\mathcal{R} \cup \mathcal{R}^c$ where 
\begin{equation*}
\mathcal{R}:=\left\{z_1, z_2, \cdots, z_l \in \mathbb{C}: |y_1|, \cdots, |y_l| \in [\eta_l, 2\eta_u] \right\},
\end{equation*}  
where form some constants $\epsilon_1>\epsilon_2>0,$ $\eta_l:=n^{-\epsilon_1} \eta_0$ and $\eta_u:=n^{-\epsilon_2} \eta_0.$ Following almost exactly the same lines of the proofs between (7.7) and (7.11) of \cite{DingWang2023}, we can conclude that 
\begin{equation}\label{eq_errorcontrol}
\int_{\mathcal{R}^c} \theta_f(z_1) \cdots \theta_f(z_l) \mathbb{E} \left[ \mathcal{Y}(z_1) \cdots \mathcal{Y}(z_l) \right] \mathrm{d}^2 z_1 \cdots \mathrm{d}^2 z_l=\mathrm{o}_{\prec}(1). 
\end{equation}   

The rest of the proof leaves to evaluate the above integral on the region $\mathcal{R}.$ Inserting the results of Lemma \ref{lem_momentresults} especially (\ref{eq_recursive}) into (\ref{eq_expectationexpansion}) and by a discussion similar to (7.13) and (7.14) of \cite{DingWang2023} and together with (\ref{eq_errorcontrol}), we can see that  
\begin{equation*}
\mathbb{E} \mathcal{T}_1^l=(l-1) \mathcal{V} \mathbb{E} \mathcal{T}_1^{l-2}+\mathrm{o}_{\prec}(1),
\end{equation*}
where we define
\begin{equation*}
\mathcal{V}:= \int_{\eta_l \leq |y_1|, |y_s| \leq 2 \eta_u}  \theta_f(z_1) \theta_f(z_s) \omega(z_1, z_s) \ \mathrm{d}^2 z_1 \ \mathrm{d}^2 z_s.
\end{equation*}

Therefore, the rest of the proof leaves to show that $\mathcal{V}=\mathsf{v}+\mathrm{o}(1).$ The discussion relies on utilizing the complex Green's theorem to translate the regional integral as in $\mathcal{V}$ into the contour integral as in $\mathsf{v}$ and then simplify the formulas by analyzing the properties of the Stieltjes transforms.  The proof  is the same as in the arguments between equations (7.15) and (7.21) and those between (8.14) and (8.17) of \cite{DingWang2023} and we omit the details. This completes the proof of (\ref{eq_needproofequationone}).

Then we discuss how to prove (\ref{eq_needproofequationtwo}).  We first introduce some notations. For $m(z)$ in (\ref{eq_simplifiednotation}), we denote $m_p(z)$ via 
\begin{equation*}
m(z)=-\frac{1-\phi}{z}+\phi m_p(z). 
\end{equation*}
Moreover, we denote the asymptotic density associated with $m_p(z)$ as $\varrho_p.$ Using the above formula and inverse formulas that $\mathrm{d}\varrho(x)=\pi^{-1} \lim_{\eta \downarrow 0} \operatorname{Im} m(x+\mathrm{i} \eta), \ \mathrm{d}\varrho_p(x)=\pi^{-1} \lim_{\eta \downarrow 0} \operatorname{Im}_p m(x+\mathrm{i} \eta),$ using the convention (\ref{eq_fxdefinition}), we can rewrite $\mathsf{M}_x$ in (\ref{mean}) as 
\begin{equation*}
\mathsf{M}_x:=p \int  f(x) \mathrm{d} \varrho_p(x). 
\end{equation*}
Consequently, similar to the representation (\ref{eq_T1integral}), we can rewrite $\mathcal{T}_2$ as follows
\begin{equation}\label{eq_T2form}
\mathcal{T}_2=\int_{\mathbb{C}} \theta_f(z) \left[\mathbb{E} \Tr \mathcal{G}(z)-pm_p(z) \right] \mathrm{d}^2 z.  
\end{equation}
Based on the above representation, we can prove the following results where the details will be given in Section \ref{sec_31proofdetails}.  
\begin{lemma}\label{lem_meancontrol}
Suppose the assumptions of Theorem \ref{thm_onesample} holds. We have that 
\begin{equation*}
\mathcal{T}_2=\frac{1}{2 \pi \mathrm{i}} \left( \int_{\mathbb{R}} f(x) b^+(x) \mathrm{d} x-\int_{\mathbb{R}} f(x) b^-(x) \mathrm{d} x \right)+\mathrm{o}_{\prec}(1), 
\end{equation*}
where $b^{\pm}(x)=\lim_{\eta \downarrow 0} b(x\pm \mathrm{i} \eta)$ with 
\begin{equation}\label{eq_definitionb(z)}
b(z):=\frac{m''(z)}{2 m'(z)}-\frac{m'(z)}{m(z)}. 
\end{equation}
\end{lemma}

Recall $f(x)$ in (\ref{eq_fxdefinition}). Using the change of variable that $\widetilde{x}=(x-\gamma)/\eta_0$ and Lemma \ref{lem_meancontrol}, we have that 
\begin{equation*}
\mathcal{T}_2=\frac{\eta_0}{2 \pi \mathrm{i}} \left( \int_{\mathbb{R}} \widetilde{x} \mathcal{K}(\widetilde{x}) b^+(\gamma+\widetilde{x} \eta_0) \mathrm{d} \widetilde{x}- \int_{\mathbb{R}} \widetilde{x} \mathcal{K}(\widetilde{x}) b^-(\gamma+\widetilde{x} \eta_0) \mathrm{d} \widetilde{x} \right)+\mathrm{o}_{\prec}(1).
\end{equation*} 
By the definition of $\mathcal{K}$ in (\ref{eq_mathcalK}), we only need to consider when $\widetilde{x}$ is bounded. Moreover,  by the arguments in the end of Section 8.2 of \cite{DingWang2023}, when $\gamma$ is in the bulk that $\gamma_-<\gamma<\gamma_+,$ for $w$ in a small neighborhood of $\gamma,$ we have that $b^{\pm}(w)=\mathrm{O}(1).$ 

Since $\eta_0=\mathrm{o}(1),$ the above discussion implies that $\mathcal{T}=\mathrm{o}(1).$  This completes the proof of (\ref{eq_needproofequationtwo}) and hence the proof of Theorem \ref{thm_onesample}.

\subsubsection{Proof details of Lemmas \ref{lem_momentresults} and \ref{lem_meancontrol}}\label{sec_31proofdetails}
In this section, we provide some details of the proof of Lemmas \ref{lem_momentresults} and \ref{lem_meancontrol}. Both lemmas have been proved in \cite{DingWang2023} assuming that $\Sigma$ is diagonal using the device of cumulant expansion (see \cite{LPCLT, yang2020linear,DingWang2023, 10.1214/20-AIHP1086, 10.1214/20-AOS1960}).  In what follows, we focus on explaining how to justify these two lemmas for general $\Sigma.$ Due to similarity, we focus on the details of Lemma \ref{lem_momentresults}. 

\begin{proof}[\bf Proof of Lemma \ref{lem_momentresults}] The proofs are similar to those of Lemma 6.2 of \cite{DingWang2023} which relies on the diagonal assumption of $\Sigma,$ i.e., $O=I$ in (\ref{eq_resolventdefinition}). In what follows, we explain how to remove this diagonal assumption; see Remark \ref{remark_diagonaliseasier} for more discussions on this aspect. For ease of the statements, we focus on the case $l=2.$  The arguments of the generalization to general $l$ are the same as those in (6.14) of \cite{DingWang2023}. For $l_1, l_2 \in \mathbb{N},$ it suffices to work with 
\begin{align*}
\mathbb{E} \left[ \mathcal{Y}(z_1)^{l_1} \mathcal{Y}(z_2)^{l_2} \right]=\sum_{i=1}^p \mathbb{E} \left[ \mathcal{Y}(z_1)^{l_1} \mathcal{Y}(z_2)^{l_2-1} (\mathcal{G}_{ii}(z_2)-\mathbb{E} \mathcal{G}_{ii}(z_2)) \right].   
\end{align*}
Note that for two random variables $a$ and $b,$ we have that $\mathbb{E}[a(1-\mathbb{E})b]=\mathbb{E}[(1-\mathbb{E})a(1-\mathbb{E})b]=\mathbb{E}[b(1-\mathbb{E})a],$ we can further rewrite the above equation as 
\begin{align*}
\mathbb{E} \left[ \mathcal{Y}(z_1)^{l_1} \mathcal{Y}(z_2)^{l_2} \right]=\sum_{i=1}^p \mathbb{E} \left\{ (1-\mathbb{E})\left[ \mathcal{Y}(z_1)^{l_1} \mathcal{Y}(z_2)^{l_2-1} \right] \mathcal{G}_{ii}(z_2) \right\}.
\end{align*}

Next we expand the above equation. Denote $Y=L\mathsf{X}, L=\Lambda^{1/2}O.$ Using the elementary identity that $z_2 \mathcal{G}(z_2)=\mathcal{G}(z_2)YY^\top-I$, by a straightforward calculation, we have that 
\begin{equation*}
z_2\mathbb{E} \left\{ (1-\mathbb{E})\left[ \mathcal{Y}(z_1)^{l_1} \mathcal{Y}(z_2)^{l_2-1} \right] \mathcal{G}_{ii}(z_2) \right\}=\mathbb{E} \left\{ (1-\mathbb{E})\left[ \mathcal{Y}(z_1)^{l_1} \mathcal{Y}(z_2)^{l_2-1} \right] \sum_{j=1}^n \sum_{k=1}^p x_{kj}(\mathcal{G}(z_2)Y)_{ij} L_{ik}  \right\}
\end{equation*}
%\begin{equation*}
%z_2\mathbb{E} \left\{ (1-\mathbb{E})\left[ \mathcal{Y}(z_1)^{l_1} \mathcal{Y}(z_2)^{l_2-1} \right] \mathcal{G}_{ii}(z_2) \right\}=\sum_{k=1}^n \sum_{j=1}^p x_{jk} L_{ij} (\mathsf{X}^\top L^\top \mathcal{G}(z_2))_{ki}.  
%\end{equation*}
Applying the cumulant expansion formula (see Proposition 3.1 of \cite{LPCLT}), we can obtain that for $1 \leq i \leq p,$ 
\begin{equation}\label{eq_elementarydecomposition}
z_2\mathbb{E} \left\{ (1-\mathbb{E})\left[ \mathcal{Y}(z_1)^{l_1} \mathcal{Y}(z_2)^{l_2-1} \right] \mathcal{G}_{ii}(z_2) \right\}=\mathsf{A}_i+\mathsf{B}_i+\mathsf{L}_i,
\end{equation}
where 
\begin{equation*}
\mathsf{A}_i:=\frac{1}{\sqrt{pn}} \sum_{j=1}^n \sum_{k=1}^p L_{ik} \mathbb{E} \left(  (1-\mathbb{E})[\mathcal{Y}(z_1)^{l_1} \mathcal{Y}(z_2)^{l_2-1}]\frac{\partial (\mathcal{G}(z_2)Y)_{ij}}{\partial x_{kj}} \right),
\end{equation*}
\begin{equation*}
\mathsf{B}_i:=\frac{1}{\sqrt{pn}} \sum_{j=1}^n \sum_{k=1}^p L_{ik} \mathbb{E} \left( \frac{\partial (1-\mathbb{E})[\mathcal{Y}(z_1)^{l_1} \mathcal{Y}(z_2)^{l_2-1}]}{\partial x_{kj}}(\mathcal{G}(z_2)Y)_{ij} \right),
\end{equation*}
\begin{equation*}
\mathsf{L}_i:= \sum_{j=1}^n \sum_{k=1}^p L_{ik} \left( \sum_{\ell=2}^4 (pn)^{-(\ell+1)/4} \frac{\kappa_{\ell+1}}{\ell!} \mathbb{E} \frac{\partial^\ell \left( \mathcal{Y}(z_1)^{l_1} \mathcal{Y}(z_2)^{l_2-1} (\mathcal{G}(z_2)Y)_{ij}\right)  }{\partial x_{kj}^\ell} \right)+\mathrm{o}(p^{-1}). 
\end{equation*}
Due to similarity, we focus our discussion on $\mathsf{A}_i.$ 
%Analogous (and easier) arguments have been made in Section 6.1 of \cite{DingWang2023} assuming $O$ is diagonal; see Remark \ref{remark_diagonaliseasier} below. 
By straightforward calculations using chain rule with $\mathcal{G}(z_2)(YY^\top-z_2)=I$, we see that  
\begin{align} \label{eq_decompositionofAi1}
\frac{\partial (\mathcal{G}(z_2)Y)_{ij}}{\partial x_{kj}}=(\mathcal{G}(z_2) L)_{ik}-(\mathcal{G}(z_2)L)_{ik} (Y^\top \mathcal{G}(z_2) Y)_{jj}-(\mathcal{G}(z_2) Y)_{ij} (L^\top \mathcal{G}Y)_{kj}.
\end{align}
Insert the above equation back into $\mathsf{A}_i,$ we have that
\begin{equation}\label{eq_Aidecomposition}
\mathsf{A}_i=\mathsf{A}_{i1}+\mathsf{A}_{i2}+\mathsf{A}_{i3},
\end{equation}
where we denote (recall $LL^\top=\Lambda$)
\begin{align*}
\mathsf{A}_{i1}:& = \frac{1}{\sqrt{pn}} \sum_{j=1}^n \sum_{k=1}^p L_{ik} \mathbb{E} \left(  (1-\mathbb{E})[\mathcal{Y}(z_1)^{l_1} \mathcal{Y}(z_2)^{l_2-1}] (\mathcal{G}(z_2) L)_{ik} \right)\\
&=\frac{1}{\sqrt{pn}} \sum_{j=1}^n  \mathbb{E} (1-\mathbb{E})[\mathcal{Y}(z_1)^{l_1} \mathcal{Y}(z_2)^{l_2-1}] (LL^\top \mathcal{G}(z_2))_{ii} \\
&=\frac{\sigma_i}{\sqrt{pn}} \sum_{j=1}^n \mathbb{E} \left\{ \mathcal{Y}(z_1)^{l_1} \mathcal{Y}(z_2)^{l_2-1}  (1-\mathbb{E})[\mathcal{G}_{ii}(z_2)] \right\}
\end{align*}
\begin{align*}
\mathsf{A}_{i2}:&=-\frac{1}{\sqrt{pn}} \sum_{j=1}^n \sum_{k=1}^p L_{ik} \mathbb{E} \left\{ \mathcal{Y}(z_1)^{l_1} \mathcal{Y}(z_2)^{l_2-1} (1-\mathbb{E})[(\mathcal{G}(z_2)L)_{ik}(Y^\top \mathcal{G}(z_2)Y)_{jj}]  \right\} \\
&=-\frac{\sigma_i}{\sqrt{pn}} \sum_{j=1}^n \mathbb{E} \left\{ \mathcal{Y}(z_1)^{l_1} \mathcal{Y}(z_2)^{l_2-1} (1-\mathbb{E})[(Y^\top \mathcal{G}(z_2) Y)_{jj} \mathcal{G}_{ii}(z_2) ] \right\} 
\end{align*}
\begin{align}\label{eq_decompositionofAi2}
\mathsf{A}_{i3}&:=-\frac{1}{\sqrt{pn}} \sum_{j=1}^n \sum_{k=1}^p L_{ik} \mathbb{E} \left\{ \mathcal{Y}(z_1)^{l_1} \mathcal{Y}(z_2)^{l_2-1} (1-\mathbb{E})[(\mathcal{G}(z_2) Y)_{ij} (L^\top \mathcal{G}Y)_{kj}]  \right\} \nonumber \\
&=-\frac{\sigma_i}{\sqrt{pn}} \sum_{j=1}^n \sum_{j=1}^n \mathbb{E} \left\{ \mathcal{Y}(z_1)^{l_1} \mathcal{Y}(z_2)^{l_2-1} (1-\mathbb{E})[(\mathcal{G}(z_1)Y)_{ij} (\mathcal{G}(z_1)Y)_{ij}] \right\}. 
\end{align}
$\mathsf{A}_{i1}$ has the same form to the left-hand side of (\ref{eq_elementarydecomposition}) so we will keep it till the end.  Moreover, $\mathsf{A}_{it}, 1 \leq t \leq 3$ have exactly the same forms as the counterparts $(\mathsf{A}_i)_t$ as in Section 6.1 of \cite{DingWang2023}. Denote 
\begin{equation*}
\mathcal{E}:=\frac{\mathrm{o}(1)}{|\eta_1|^{l_1} |\eta_2|^{l_2}}.
\end{equation*}
Together with Theorem \ref{thm_locallaw}, we find that the results in Section 6.1 of \cite{DingWang2023} still hold in our setting so that 
\begin{align*}
\mathsf{A}_{i2} &= \sigma_i \phi^{1/2}(-z_2 m_p(z_2)-1) \mathbb{E} \left[ \mathcal{Y}(z_1)^{l_1} \mathcal{Y}(z_2)^{l_2-1} (1-\mathbb{E})[\mathcal{G}_{ii}(z_2)] \right] \\
&+\frac{\sigma_i \phi^{1/2}}{1+\phi^{-1/2} m(z_2) \sigma_i} \mathbb{E} \left[ \mathcal{Y}(z_1)^{l_1} \mathcal{Y}(z_2)^{l_2-1} (1-\mathbb{E}) \left[p^{-1} \operatorname{Tr} \mathcal{G}(z_2) \right] \right]+\mathrm{O}_{\prec}(\mathcal{E}).
\end{align*} 
\begin{equation*}
\mathsf{A}_{i3}=\mathrm{O}_{\prec} \left( \mathcal{E}  \right).
\end{equation*}
Similarly, we can handle $\mathsf{B}_i$ and $\mathsf{L}_i.$ Especially, with analogous manipulation to the arguments between (\ref{eq_decompositionofAi1}) and (\ref{eq_decompositionofAi2}), we have that 
\begin{align*}
-\sum_{i=1}^p & \frac{\mathsf{B}_i}{z_2(1+\sigma_i \phi^{-1/2} m(z_2))}=\sum_{i=1}^p \frac{2\sigma_i}{\sqrt{pn} (1+\sigma_i \phi^{-1/2} m(z_2))} \mathbb{E} \Big\{ l_1 \mathcal{Y}(z_1)^{l_1-1} \mathcal{Y}(z_2)^{l_2-1}  \\
& \times \frac{\partial}{\partial z_1} \left( \frac{z_1/z_2(1+\phi^{-1/2} m(z_2) \sigma_i)^{-1}-(1+\phi^{-1/2} m(z_1) \sigma_i)^{-1}}{z_1-z_2}\right) \\
&+\frac{l_2-1}{2}  \mathcal{Y}(z_1)^{l_1} \mathcal{Y}(z_2)^{l_2-2} \frac{\partial}{\partial z_2} \left( -(z_2(1+\phi^{-1/2}m(z_2) \sigma_i))^{-1}+z_2 \left[ (z_2(1+\phi^{-1/2}m(z_2) \sigma_i))^{-1} \right]'  \right) \Big\} \\
&+\mathrm{O}_{\prec}(\mathcal{E}),
\end{align*} 
\begin{align*}
\sum_{i=1}^p & \frac{\mathsf{L}_i}{z_2(1+\sigma_i \phi^{-1/2} m(z_2))}= p^{-1} \sum_{i=1}^p \frac{-\kappa_4 \sigma_i^2 }{z_2(1+\sigma_i \phi^{-1/2} m(z_2))} \mathbb{E} \Big\{ \frac{m(z_2)}{1+\phi^{-1/2} m(z_2) \sigma_i} l_1 \mathcal{Y}(z_1)^{l_1-1} \mathcal{Y}(z_2)^{l_2-1} \\
& \times \frac{\partial }{\partial z_1} \frac{m(z_1)}{1+\phi^{-1/2} m(z_1) \sigma_i}+\frac{l_2-1}{2} \mathcal{Y}(z_1)^{l_1} \mathcal{Y}(z_2)^{l_2-2} \frac{\partial}{\partial z_2} \left( \frac{m(z_2)}{1+\phi^{-1/2} m(z_2) \sigma_i} \right)^2 \Big\}+\mathrm{O}_{\prec} \left( \mathcal{E}\right).  
\end{align*}

Inserting all the above equations back into (\ref{eq_elementarydecomposition}) and using Lemma 5.2 of \cite{DingWang2023}, with a discussion similar to (6.13) of \cite{DingWang2023}, we have that
		\[ \mathbb{E}[\mathcal{Y}^{l_1}(z_1) \mathcal{Y}^{l_2}(z_2)] = l_1 \mathbb{E}[\mathcal{Y}^{l_1-1}(z_1) \mathcal{Y}^{l_2-1}(z_2)] \mathcal{W}_1(z_1, z_2) + (l_2-1) \mathbb{E}[\mathcal{Y}^{l_1}(z_1) \mathcal{Y}^{l_2-2}(z_2)] \mathcal{W}_2(z_1, z_2) + \mathrm{O}_{\prec}(\mathcal{E}), \]
where $\mathcal{W}_k(z_1, z_2), k=1,2,$ are defined as 
\begin{align*}
\mathcal{W}_1(z_1, z_2):&=\left(1-\frac{1}{z_2 p} \sum_{i=1}^p \frac{\phi^{1/2} \sigma_i}{(1+\phi^{-1/2} m(z_2) \sigma_i)^2} \right)^{-1} \sum_{i=1}^p \frac{1}{z_2(1+\sigma_i \phi^{-1/2} m(z_2))} \Big\{-\frac{2 \sigma_i}{\sqrt{pn}} \\
&\times \frac{\partial}{\partial z_1} \left( \frac{z_1/z_2(1+\phi^{-1/2} m(z_2) \sigma_i)^{-1}-(1+\phi^{-1/2} m(z_1) \sigma_i)^{-1}}{z_1-z_2}\right) \\
&-\frac{\kappa_4 \sigma_i^2}{p} \frac{m(z_2)}{1+\phi^{-1/2} m(z_2) \sigma_i} \times \frac{\partial }{\partial z_1} \frac{m(z_1)}{1+\phi^{-1/2} m(z_1) \sigma_i} \Big \},
\end{align*} 
\begin{align*}
\mathcal{W}_2(z_1, z_2):& =\left(1-\frac{1}{z_2 p} \sum_{i=1}^p \frac{\phi^{1/2} \sigma_i}{(1+\phi^{-1/2} m(z_2) \sigma_i)^2} \right)^{-1}\sum_{i=1}^p \frac{1}{z_2(1+\sigma_i \phi^{-1/2} m(z_2))} \Big\{-\frac{2 \sigma_i}{\sqrt{pn}} \\
&  \times \frac{1}{2} \frac{\partial}{\partial z_2} \left( -(z_2(1+\phi^{-1/2}m(z_2) \sigma_i))^{-1}+z_2 \left[ (z_2(1+\phi^{-1/2}m(z_2) \sigma_i))^{-1} \right]'  \right) \\
& -\frac{\kappa_4 \sigma_i^2 }{p} \times \frac{\partial}{\partial z_2} \left( \frac{m(z_2)}{1+\phi^{-1/2} m(z_2) \sigma_i} \right)^2 \Big\}.  
\end{align*}
Furthermore, together with the results between equations (6.13) and (6.14) of \cite{DingWang2023},  we further find that 
\begin{align*}
\mathcal{W}_k(z_1, z_2):=\omega(z_1, z_2), \ k=1,2.  
\end{align*} 
This yields that 
		\[ \mathbb{E}[\mathcal{Y}^{l_1}(z_1) \mathcal{Y}^{l_2}(z_2)] = l_1 \mathbb{E}[\mathcal{Y}^{l_1-1}(z_1) \mathcal{Y}^{l_2-1}(z_2)] \omega(z_1, z_2) + (l_2-1) \mathbb{E}[\mathcal{Y}^{l_1}(z_1) \mathcal{Y}^{l_2-2}(z_2)] \omega(z_1, z_2) + \mathrm{O}_{\prec}(\mathcal{E}). \]
Then by Wick's probability theorem (for example, see Lemma 5.11 of \cite{DingWang2023}), we can complete the proof. 
\end{proof}

\begin{remark}\label{remark_diagonaliseasier}
We mention that compared to the proof of Lemma 6.2 of \cite{DingWang2023} which only handles the diagonal $\Sigma$, the main difference lies in (\ref{eq_elementarydecomposition}). The counterpart of the diagonal case, i.e., (6.3) of \cite{DingWang2023} only involves a single summation which makes the calculation easier.  
\end{remark}		
		
Next, we prove Lemma \ref{lem_meancontrol}. Similar results have been proved in Section 7.2 of \cite{DingWang2023} assuming that $O=I$ in (\ref{eq_othernotations}). We will focus the explanation on how to remove this assumption. In the fact, the arguments are similar to those in the proof of Lemma \ref{lem_momentresults} and Remark \ref{remark_diagonaliseasier}, we only sketch the idea. 
		
\begin{proof}[\bf Proof of Lemma \ref{lem_meancontrol}] In light of (\ref{eq_T2form}), it suffices to work with $\mathcal{G}_{ii}(z).$ Similar to the discussion of (\ref{eq_elementarydecomposition}), we can decompose that 
\begin{align}\label{eq_originaldefinition}
z\mathbb{E}\left[\mathcal{G}_{ii}(z)\right]&=
 \mathbb{E} \sum_{j=1}^n \sum_{k=1}^p x_{kj}(\mathcal{G}(z)Y)_{ij} L_{ik}-1  \\
&=\mathsf{S}_{i1}+\mathsf{S}_{i2}+\mathrm{o}_{\prec}(p^{-1}), \nonumber
\end{align}
where we denote
\begin{equation}\label{eq_Si1}
\mathsf{S}_{i1}:=\frac{1}{\sqrt{pn}} \sum_{j=1}^n \sum_{k=1}^p L_{ik} \frac{\partial (\mathcal{G}(z) Y)_{ij}}{\partial x_{kj}} -1,
\end{equation}
\begin{equation*}
\mathsf{S}_{i2}:= \sum_{j=1}^n \sum_{k=1}^p L_{ik} \left( \sum_{\ell=2}^4 (pn)^{-(\ell+1)/4} \frac{\kappa_{\ell+1}}{\ell!} \mathbb{E} \frac{\partial^\ell \left((\mathcal{G}(z)Y)_{ij}\right)  }{\partial x_{kj}^\ell} \right).
\end{equation*}

For $\mathsf{S}_{i1},$ with a discussion similar to $\mathsf{A}_i$ in (\ref{eq_Aidecomposition}), using Theorem \ref{thm_locallaw}, we can show that (see the arguments between equations (7.25) and (7.27) of \cite{DingWang2023})
\begin{align*}
\mathsf{S}_{i1}&=\sigma_i(\phi^{-1/2}-\phi^{1/2}) \mathbb{E} \mathcal{G}_{ii}(z)-z\sigma_i \phi^{1/2}m_p(z) \mathbb{E} \left( \mathcal{G}_{ii}(z)+\frac{1}{z(1+\phi^{-1/2} m(z) \sigma_i)} \right) \\
&+\frac{\sigma_i \phi^{1/2}}{p(1+\phi^{-1/2} m(z) \sigma_i)} \mathbb{E} \operatorname{Tr} \mathcal{G}(z)-\frac{\sigma_i^2}{p} \frac{m'(z)}{(1+\phi^{-1/2} \sigma_i m(z))^2}-1+\mathrm{o}_{\prec}(p^{-1}). 
\end{align*}
Insert the above equation back into (\ref{eq_originaldefinition}) and summarize over the index $i$ from $1$ to $p$, with a discussion similar to (7.27) of \cite{DingWang2023}, we can obtain that 
\begin{align*}
\mathbb{E} \operatorname{Tr} \mathcal{G}(z)& =\sum_{i=1}^p \frac{1}{z(1+\sigma_i \phi^{-1/2} m(z))} \Big( \frac{\sigma_i \phi^{1/2}}{p(1+\phi^{-1/2} m(z) \sigma_i)} \mathbb{E}[\operatorname{Tr} \mathcal{G}(z)-p m_p(z) ] \\
&-\frac{\sigma_i^2 m'(z)}{p(1+\phi^{-1/2} \sigma_i m(z))^2}-1+\mathsf{S}_{i2} \Big)+\mathrm{o}_{\prec}(1).  
\end{align*}
Following lines of (7.33) of \cite{DingWang2023} and using Lemma 5.2 therein, we can obtain that 
\begin{align*}
&\left(1-\frac{1}{p} \sum_{i=1}^p \frac{\phi^{1/2} \sigma_i}{z(1+\phi^{-1/2} m(z) \sigma_i)^2} \right) \mathbb{E}[\operatorname{Tr} \mathcal{G}(z)-pm_p(z) ] \\
&=-\frac{1}{p}\sum_{i=1}^p \frac{\sigma_i^2 m'(z)}{z(1+\phi^{-1/2} m(z) \sigma_i)^3}+\sum_{i=1}^p \frac{\mathsf{S}_{i2}}{z(1+\sigma_i \phi^{-1/2} m(z))}+\mathrm{o}_{\prec}(1). 
\end{align*}
Recall $b(z)$ in (\ref{eq_definitionb(z)}). By (7.33) and (7.34) of \cite{DingWang2023}, we can further conclude that 
\begin{equation*}
\mathbb{E}[\operatorname{Tr} \mathcal{G}(z)-pm_p(z)]=b(z)+\sum_{i=1}^p \frac{\mathsf{S}_{i2}}{z(1+\sigma_i \phi^{-1/2} m(z))}+\mathrm{o}_{\prec}(1).
\end{equation*}

Then we can follow exactly the same  arguments between (7.30) and (7.34) of \cite{DingWang2023} to conclude that 
\begin{equation*}
\int_{\mathbb{C}}\sum_{i=1}^p \frac{\mathsf{S}_{i2}}{z(1+\sigma_i \phi^{-1/2} m(z))} \mathrm{d} z=\mathrm{o}_{\prec}(1). 
\end{equation*}
Together with (\ref{eq_T2form}), we can obtain
\begin{equation*}
\mathcal{T}_2=\int_{\mathbb{C}} \theta_f(z) b(z) \mathrm{d}^2 z.  
\end{equation*}
Then by a discussion similar to the arguments below (\ref{eq_recursive}), we can conclude our proof. 
 
%$$\begin{aligned}
%			&z\mathbb{E}\left[\mathcal{G}_{i i}\right]\\
%			= & \mathbb{E}\sum_{j k} X_{j k}\left(T E_{j k} X^{\top} T^{\top} \mathcal{G}\right)_{i i}-1=\mathbb{E}\sum_{j k} X_{j k}\left(T\right)_{i j}\left(X^{\top}T^{\top} \mathcal{G}\right)_{k i}-1 \\
%			= & \sum_{j k} \frac{c_{j k}^{(2)}\left(T\right)_{i j}}{\sqrt{np}} \mathbb{E} \frac{\partial\left(X^{\top} T^{\top} \mathcal{G}\right)_{k i}}{\partial X_{j k}}-1  +\sum_{l\ge 2}\sum_{j k} \frac{c_{j k}^{(l+1)}\left(T\right)_{i j}}{l !(np)^{(l+1)/ 4}} \mathbb{E} \frac{\partial^l\left(X^{\top} T^{\top} \mathcal{G}\right)_{k i}}{\partial X_{j k}^l} \\
%			= & \frac{1}{\sqrt{np}} \sum_{j k}\left(T\right)_{i j} \mathbb{E}\left[\left(T^{\top} \mathcal{G}\right)_{j i}-\left(X^{\top} T^{\top} \mathcal{G} T\right)_{k j}\left(X^{\top}T^{\top} \mathcal{G}\right)_{k i}-\left(X^{\top} T^{\top} \mathcal{G} T X\right)_{k k}\left(T^{\top}\mathcal{G}\right)_{j i}\right]-1\\
%			&+\sum_{l\ge 2}\sum_{j k} \frac{c_{j k}^{(l+1)}\left(T\right)_{i j}}{l !(np)^{(l+1)/ 4}} \mathbb{E} \frac{\partial^l\left(X^{\top} T^{\top} \mathcal{G}\right)_{k i}}{\partial X_{j k}^l} 
%		\end{aligned}$$

\end{proof}

\section{Automated tuning parameter selection}\label{suppl_tuningparameterselection}

In the section, we discuss how to choose the tuning parameters $\{\eta_0^i\}$ in Step three of Algorithm \ref{alg:bootstrapping} and the control threshold $\delta$ in the last step of Algorithm \ref{alg:bootstrapping}.

We start with $\{\eta_0^i\}.$ According to our theoretical results in Corollary \ref{cor_twosample}, $\eta_0$ should satisfy that $n^{-1} \ll \eta_0 \ll 1.$ Moreover, by the construction in (\ref{eq_2sampleindividualstatdef}), $\eta_0$ can be understood as the bandwidth which controls the number of eigenvalues to be used in our test procedure. Based on such a connection with the bandwidth, inspired by bandwidth selection procedure in \cite{levine2006bandwidth}, in Step three of Algorithm \ref{alg:bootstrapping}, for $i \in \mathcal{S}_1,$ we will choose 
\begin{equation}\label{eq_finalchoice}
\eta_0^i=\theta \times \operatorname{std}(\{\gamma_j^i\}). 
\end{equation} 
Here $\operatorname{std}(\{\gamma_j^i\})>0$ is the standard derivation of the sequence $\{\gamma_j^i\}$ and $\theta \gg n^{-1}$ is some (small) constant. Therefore, the choices of the sequence of the tuning parameters reduce to the selection of a global tuning parameter $\theta.$ In Algorithm \ref{algo2}, we provide a smoothing based method to choose $\theta$ so that our Algorithm \ref{alg:bootstrapping} can be stabilized (i.e., the decision ratio in (\ref{eq_decisionrationdefinition}) is  large under the alternatives and will not change significantly for small changes in $\eta_0$).

	\begin{algorithm}[htp]
		\caption{\bf Parameter selection for (\ref{eq_finalchoice})}
		\label{algo2}
		\normalsize
		\begin{flushleft}
		\noindent{\bf Inputs:} Inputs of Algorithm \ref{alg:bootstrapping} and a pre-specified set of $s$ grids of positive values $c_1<c_2<\cdots<c_s.$

			\noindent{\bf Step one:} For $k=1,\cdots,s,$ execute steps one through five of Algorithm \ref{alg:bootstrapping} by setting $\eta_0^i$ in the form of (\ref{eq_finalchoice}) with $\theta=c_k.$ Record all the $s$ decision ratios (c.f. (\ref{eq_decisionrationdefinition})) as $\{\mathsf{DR}_k\}$. 

			\noindent{\bf Step two:} Compute the simple moving average (SMA) of $\{\mathsf{DR}_k\}$ with window size $3$ and denote the sequence as $\{\mathsf{DR}'_k\},$ i.e., for $1 \leq k \leq s-2,$ $$\mathsf{DR}_k^{'} = \frac{\mathsf{DR}_{k}+\mathsf{DR}_{k+1}+\mathsf{DR}_{k+2}}{3}.$$ 
			
			\noindent{\bf Step three:} For all $1\leq t \leq s-3,$ compute the variance of $\{\mathsf{DR}_k\}_{k=1}^{t+1}$  and denote them as $\{\mathsf{v}_t\}_{t=1}^{s-3}.$ 
			
			\noindent{\bf Step four:} For all $3 \leq \ell \leq s-2,$ select $\ell$ so that the following two conditions are satisfied 
			\begin{equation}\label{eq_importanttwocondistions}
			 \mathsf{DR}_\ell>\frac{\max_k \mathsf{DR}_k}{5}, \ \  \mathsf{v}_{\ell-2}>\mathsf{v}_{\ell-1}.
			\end{equation}
			 	
			\noindent{\bf Output:} Choose $\theta=c_{\ell}.$
			
		\end{flushleft}
	\end{algorithm}

Then we explain how to choose the threshold $\delta$ as discussed in Remark \ref{rmk_tuning} using a calibration procedure following \cite{9779233}. The details are provided in Algorithm \ref{algo3}.

\begin{algorithm}[htp]
	\caption{\bf Tuning parameter selection }
	\label{algo3}
	\normalsize
	\begin{flushleft}
		\noindent{\bf Inputs:} $n_1$, $n_2$, $n$, $p$ and $\mathsf{K}$ and type I error $\alpha$ from Algorithm \ref{alg:bootstrapping}.
		
		\noindent{\bf Step one:} For some large value $\mathsf{B}$ (say $\mathsf{B}=1,000$), generate a sequence of $\mathsf{B}$ $p$-dimensional standard multivariate Gaussian samples $\mathcal{X}_i$'s and $\mathcal{Y}_i$'s of sizes $n_1$ and $n_2$.
		
		 \noindent{\bf Step two:} Perform Algorithm \ref{alg:bootstrapping} on $\mathcal{X}_i$ and $\mathcal{Y}_i$ with parameter $n$, $\alpha$ and $\mathsf{K}$. Record the decision ratios as $\mathsf{DR}_i, \ 1 \leq i \leq \mathsf{B}.$
		
		\noindent{\bf Step three:} Compute the $(1-\alpha)$-quantile of $\{\mathsf{DR}_i\}_{i=1}^{\mathsf{B}}$, denoted as $\delta,$ following 
		\begin{equation*}
		\frac{\# \{\mathsf{DR}_i \leq \delta\}}{\mathsf{B}} \geq 1-\alpha.
		\end{equation*}
		
		\noindent{\bf Output:} The threshold $\delta.$
		
	\end{flushleft}
\end{algorithm}

%\begin{remark}
%{\color{blue} Use either CV or KDE based method. The intuition behind is clear. If $c$ is small, nothing are selected so that we always fail to reject. If $c$ is too large that the variance change to $2/\eta_0^2,$ fail to reject either. } 
%\end{remark}
		
%		\section{Additional numerical results}
		
%chicago		
\bibliographystyle{plain}	
\bibliography{lsft}
%\begin{thebibliography}{100}
%
%\bibitem{alain2014regularized}
%Guillaume Alain and Yoshua Bengio.
%\newblock {\em What regularized auto-encoders learn from the data-generating distribution.}
%\newblock {\bf Journal of Machine Learning Research}, 15(1):3563--3593,
%  2014.
%  
%  \end{thebibliography}		

%	
	\end{document}